\documentclass[11pt,authoryear]{article}

\usepackage[margin=1in]{geometry}
\usepackage[utf8]{inputenc}
\usepackage[T1]{fontenc}
\usepackage{amsmath,amssymb,amsthm,mathtools}
\usepackage{amsfonts}
\usepackage{algorithm}
\usepackage{algpseudocode}
\usepackage{booktabs}
\usepackage{graphicx}
\usepackage{longtable}
\usepackage[numbers]{natbib}
\usepackage{xcolor}
\usepackage{url}
\usepackage[colorlinks=true,citecolor=blue,linkcolor=blue,urlcolor=blue,hypertexnames=false]{hyperref}
\usepackage{authblk}

\hypersetup{
    pdftitle={Approximation algorithms for the prize-collecting rural postman problem}, pdfauthor={Hong Li, Jianping Li, Wei Li, Runtao Xie, Xiaoxiao Yang}, pdfkeywords={arc routing, prize-collecting routing, approximation algorithms, linear programming}
}

\allowdisplaybreaks[1]
\newtheorem{theorem}{Theorem}[section]
\newtheorem{lemma}[theorem]{Lemma}

\theoremstyle{definition}

\theoremstyle{remark}

\title{Approximation algorithms for the prize-collecting rural postman problem\thanks{This work was supported by the National Natural Science Foundation of China (No. 12361066) and the Project of Yunling Scholars Training of Yunnan Province (No. K264202011820).}}

\author[1]{Hong Li\thanks{Email: \texttt{honglimath@126.com}.}}
\author[1]{Jianping Li\thanks{Corresponding author. Email: \texttt{jianping@ynu.edu.cn}.}}
\author[1]{Wei Li}
\author[1]{Runtao Xie}
\author[1]{Xiaoxiao Yang}

\affil[1]{School of Mathematics and Statistics, Yunnan University}
\date{}

\begin{document}
\maketitle

\begin{abstract}
In this paper, we study the prize-collecting rural postman problem (PCRPP), a variant of the rural postman problem. In an instance of the PCRPP, one is given an undirected graph whose edges have nonnegative lengths and nonnegative profits, together with a specified root vertex. The goal is to find a closed walk that starts and ends at the root vertex and minimizes the sum of the walk length and the profits of all edges that the walk does not traverse. A natural way to design an approximation algorithm for the PCRPP is to construct a prize-collecting traveling salesman problem (PCTSP) instance from the given PCRPP instance, apply an approximation algorithm to the PCTSP instance, and then convert the resulting solution to the PCTSP instance into a solution to the PCRPP instance. We show that this approach has an inherent factor-two barrier: even if the constructed PCTSP instance is solved exactly, the resulting solution to the PCRPP instance can have objective value arbitrarily close to twice the optimum value of the PCRPP instance. Our main result is a polynomial time approximation algorithm with an approximation ratio strictly smaller than \(1.6\) for the PCRPP. On a public benchmark set of 118 instances, the proposed algorithm has average and maximum optimality gaps of \(3.39\%\) and \(12.12\%\), respectively.
\end{abstract}


\section{Introduction}\label{sec:intro}

The rural postman problem is a classical arc routing problem \citep{Orloff1974}. In this problem, we are given an undirected graph with nonnegative edge lengths and a set of required edges, and the task is to find a minimum-length closed walk that traverses all required edges. A natural variant arises when the set of required edges is no longer prescribed in advance. Instead, each edge has a profit that is collected once the edge is traversed; one must choose which edges to traverse based on the tradeoff between their profits and the additional travel length needed to include them in a closed walk. \citet{Araoz2006} introduced this model under the name privatized rural postman problem, and \citet{Araoz2009} later studied the same model explicitly as the prize-collecting rural postman problem (PCRPP).

In the PCRPP, we are given an undirected graph in which each edge has a nonnegative length and a nonnegative profit, together with a specified root vertex \(r\). The problem is to find a closed walk, not necessarily simple, that starts and ends at \(r\), such that the profit of an edge is collected when the walk traverses this edge at least once. The objective is to maximize the collected profit minus the walk length. Note that the profit of an edge can be collected at most once, whereas the walk length counts an edge as many times as it is traversed. Since feasible solutions may have negative objective value and the optimal value of this maximization formulation may be zero, this formulation is not well suited to multiplicative approximation analysis. We therefore use the corresponding minimization version in this paper. In this version, the objective is to minimize the walk length plus the total profit that is not collected by the walk, i.e., the walk length plus the total profit of the edges that it does not traverse. For the same instance, the two formulations have the same optimal solutions, since, for every feasible solution, the minimization objective value equals the total edge profit of the instance minus the maximization objective value.

\subsection{Mathematical model}
\label{sec:model}

Consider a PCRPP instance \(G=(V,E;w,p;r)\), where \(r\in V\) is the root vertex, \(w_e\ge 0\) denotes the length of edge \(e\), and \(p_e\ge 0\) denotes its profit. We may assume that \(G\) is connected. Indeed, no closed walk starting and ending at \(r\) can traverse an edge in a connected component not containing \(r\). Thus it suffices to restrict the instance to the connected component containing \(r\); under the minimization objective, this changes the objective value of every feasible solution by the same fixed constant and does not affect the approximation analysis.

We present an integer programming formulation that explicitly describes the structure of optimal solutions under the minimization objective. Let \(x_e\) denote the number of times edge \(e\) is traversed, let \(z_e\in\{0,1\}\) indicate whether the profit of \(e\) is collected, let \(y_v\in\{0,1\}\) indicate whether vertex \(v\) is visited, and let \(k_v\in\mathbb Z_{\ge0}\) denote the number of times vertex \(v\) is visited. Because all edge lengths are nonnegative and each edge profit is collected at most once, there exists an optimal solution with \(x_e\in\{0,1,2\}\) for every edge \(e\). Thus the following model is sufficient to capture an optimal solution to the PCRPP instance:

\begin{align}
\min \quad &
\sum_{e\in E} w_e x_e + \sum_{e\in E} p_e (1-z_e)
\label{eq:obj_min}\\
\text{\textit{s.t.}} \quad
& x(\delta(v)) = 2k_v,
&& \forall\, v\in V,
\label{eq:parity} \\
& x(\delta(S)) \ge 2y_v,
&& \forall\, S\subseteq V\setminus\{r\},\ \forall\, v\in S,
\label{eq:conn} \\
& z_e \le x_e \le 2z_e,
&& \forall\, e\in E,
\label{eq:usage} \\
& x_e \le 2y_u,\ x_e \le 2y_v,
&& \forall\, e=uv\in E,
\label{eq:visit} \\
& z_e \le y_u,\ z_e \le y_v,
&& \forall\, e=uv\in E,
\label{eq:couple} \\
& y_r = 1,\ z_e,y_v\in\{0,1\},\
x_e\in\{0,1,2\},\
k_v\in\mathbb Z_{\ge0}.
\label{eq:vars}
\end{align}

In \eqref{eq:obj_min}, the first term equals the walk length, and the second term equals the total profit of the edges whose profits are not collected. Constraint~\eqref{eq:parity} enforces that all vertex degrees are even. Constraint~\eqref{eq:conn} is a rooted subtour-elimination constraint: it ensures edge connectivity between every visited vertex and the root vertex \(r\). Constraint~\eqref{eq:usage} enforces \(z_e=1\) exactly when \(x_e\ge1\). Constraint~\eqref{eq:visit} ensures that the endpoints of every traversed edge are visited, and constraint~\eqref{eq:couple} ensures that the profit of an edge can be collected only if both endpoints are visited. This formulation is included for structural and expository purposes only; it is not the formulation used in our approximation algorithm.

To relate the minimization objective in \eqref{eq:obj_min} to the original maximization formulation of the PCRPP, recall that \citet{Araoz2006,Araoz2009} use the objective

\begin{equation}
\max \sum_{e\in E} p_e z_e - \sum_{e\in E} w_e x_e .
\label{eq:obj_araoz}
\end{equation}
For every feasible solution,
\[
\sum_{e\in E} w_e x_e + \sum_{e\in E} p_e(1-z_e)
=
\sum_{e\in E} p_e
-
\left(
\sum_{e\in E} p_e z_e
-
\sum_{e\in E} w_e x_e
\right).
\]
Since the sum of all edge profits is fixed for the instance, minimizing \eqref{eq:obj_min} and maximizing \eqref{eq:obj_araoz} produce exactly the same optimal solutions. However, this transformation does not preserve multiplicative approximation ratios. Throughout the paper, we state all approximation ratios for this minimization objective~\eqref{eq:obj_min}.

\subsection{Related work}
\label{sec:related}

The PCRPP belongs to the broader class of arc routing problems with profits. \citet{Araoz2006} introduced the problem under the name privatized rural postman problem, proved its NP-hardness, and gave polynomial time exact algorithms for several special cases. \citet{Araoz2009} later studied the same model as the PCRPP and proposed exact and heuristic algorithms for the PCRPP. \citet{Palma2011} proposed a tabu-search heuristic algorithm for the PCRPP. Other closely related arc routing problems with profits include the clustered prize-collecting arc routing problem \citep{AraozFranquesa2009}, the windy clustered prize-collecting arc routing problem \citep{CorberanFernandezFranquesaSanchis2011}, the time-dependent prize-collecting arc routing problem \citep{BlackEgleseWohlk2013}, the directed profitable rural postman problem \citep{ArchettiGuastarobaSperanza2014,ColombiMansini2014}, and the profitable windy rural postman problem \citep{Avila2016}. For a broader account of arc routing with profits, we refer the reader to \citet{ArchettiSperanza2015}. Most work on these models has focused on exact and heuristic algorithms.

Another related model is the restricted Chinese postman problem with penalties, which can be viewed as the minimization version of the PCRPP with the additional requirement that all vertices be visited. \citet{ZhuPan2021} studied this problem and gave a \(1.5\)-approximation algorithm. More recently, \citet{PanZhu2024} studied approximation algorithms for multiple-vehicle extensions of the restricted Chinese postman problem with penalties. These models are different from the PCRPP considered in this paper: the PCRPP imposes no requirement that all vertices be visited. Thus the above approximation algorithms do not apply to the PCRPP.

Our approximation analysis is closely connected to the prize-collecting traveling salesman problem (PCTSP). We use the standard minimization form of the PCTSP: given a complete graph with nonnegative edge lengths \(w\) satisfying the triangle inequality, a root vertex \(r\), and a nonnegative penalty \(\pi_v\) on each vertex \(v\), the goal is to find a closed walk that starts and ends at the root vertex \(r\) and minimizes the walk length plus the total penalty of the vertices not visited by the walk. This objective has a length-plus-penalty form similar to the minimization version of the PCRPP, but the two problems differ in where the penalties or profits are placed: the PCTSP assigns penalties to vertices, whereas the PCRPP assigns profits to edges.

The PCTSP has been studied mainly from the viewpoint of approximation algorithms. \citet{Balas1989} introduced the PCTSP. \citet{Bienstock1993} used threshold rounding to obtain a \(2.5\)-approximation algorithm. \citet{GoemansWilliamson1995} developed a primal-dual method and obtained a \(2\)-approximation algorithm. \citet{Goemans2009} combined threshold rounding with a refined analysis of the primal-dual method and obtained a \(1.915\)-approximation algorithm. More recently, \citet{Blauth2023} used splitting-off techniques, a tree decomposition, and parity correction to obtain a \(1.774\)-approximation algorithm. Building on this line of work, \citet{Blauth2026} combined splitting-off techniques and a tree decomposition with pruning and parity correction to obtain a better-than-\(1.6\)-approximation algorithm.

\subsection{Our contributions}
\label{sec:contributions}

Previous studies on the PCRPP and related profitable arc routing problems have mainly focused on exact and heuristic algorithms. We are not aware of prior work that systematically studies polynomial time approximation algorithms with worst-case guarantees for the PCRPP.

A natural way to design an approximation algorithm for the PCRPP is to construct a PCTSP instance from the given PCRPP instance, apply an approximation algorithm to the PCTSP instance, and then convert the resulting solution to the PCTSP instance into a solution to the PCRPP instance. We refer to this idea as the reduction-based approach. The reduced PCTSP instance is constructed as follows. We first form a subdivided graph from the original graph. For each edge \(e=uv\) with \(p_e>0\), introduce a representative vertex \(s_e\), replace \(e\) by two edges \(us_e\) and \(s_ev\), each of length \(w_e/2\), and assign penalty \(\pi_{s_e}=p_e\) to \(s_e\). Each edge \(e\) with \(p_e=0\) is kept with its original length. We then construct a complete graph on the vertex set \(R=\{r\}\cup\{s_e:p_e>0\}\). For any two distinct vertices \(a,b\in R\), the length of the edge \(ab\) is the shortest path length between \(a\) and \(b\) in the subdivided graph. These edge lengths satisfy the triangle inequality by construction. Let \(G^R\) denote this complete graph. The reduced PCTSP instance is defined on \(G^R\), with root vertex \(r\), penalty \(\pi_{s_e}=p_e\) for every representative vertex \(s_e\), and \(\pi_r=0\). Given a solution to the reduced PCTSP instance, the reduction-based approach interprets every visited representative vertex \(s_e\) as selecting the original edge \(e\), and then constructs a closed walk that traverses all selected original edges. However, this approach has an inherent factor-two barrier: even if the reduced PCTSP instance is solved exactly, it cannot guarantee an approximation ratio strictly below \(2\) for the PCRPP. Figure~\ref{fig:reduction-barrier} gives a simple instance illustrating this barrier.

\begin{figure}[htbp]
 \centering
 \includegraphics[width=0.82\textwidth]{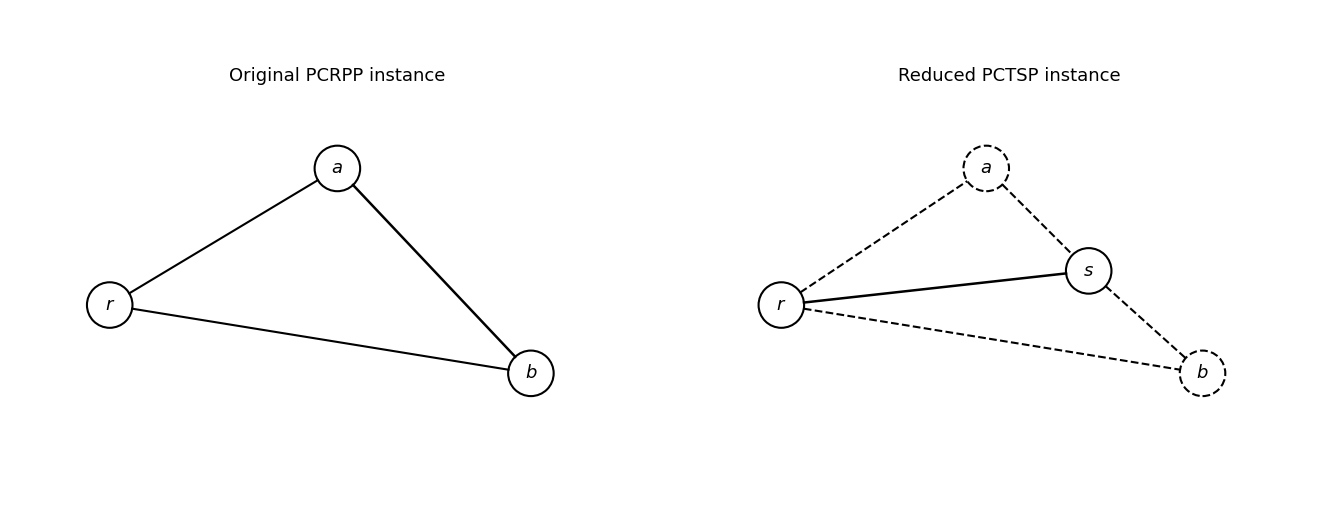}
 \caption{A factor-two barrier for the natural reduction-based approach from the PCRPP to the PCTSP. In the PCRPP instance, the edge lengths are \(w_{ra}=\varepsilon\), \(w_{ab}=1\), and \(w_{rb}=1\), where \(0<\varepsilon<1/2\). The only positive-profit edge is \(ab\), with \(p_{ab}=1+3\varepsilon\). The rooted closed walk containing only \(r\) has objective value \(1+3\varepsilon\), and it is optimal for the PCRPP. In the reduced PCTSP instance, the edge \(ab\) is represented by a vertex \(s\) with penalty \(\pi_s=1+3\varepsilon\). The length of the edge \(rs\) in the complete metric graph is \(1/2+\varepsilon\), so the walk \(r\to s\to r\) has objective value \(1+2\varepsilon\), which is smaller than the penalty \(\pi_s=1+3\varepsilon\). Hence the optimal solution to the PCTSP instance visits \(s\). Any feasible solution to the PCRPP instance that traverses \(ab\) has objective value at least \(2+\varepsilon\), achieved by \(r\to a\to b\to r\). Therefore the ratio between the objective value of any such feasible solution to the PCRPP instance and the optimal value of the PCRPP instance is at least \((2+\varepsilon)/(1+3\varepsilon)\), which tends to \(2\) as \(\varepsilon\to0\).}
 \label{fig:reduction-barrier}
\end{figure}

For comparison in the computational experiments, we also instantiate this reduction-based approach as a PCTSP-reduction algorithm in Appendix~\ref{app:pctsp-reduction}. The PCTSP-reduction algorithm is a \(2\rho\)-approximation algorithm for the PCRPP, where \(\rho\) denotes the approximation ratio of the PCTSP algorithm applied to the reduced PCTSP instance. The details of the algorithm and its approximation-ratio analysis are given in Appendix~\ref{app:pctsp-reduction}. In particular, using the better-than-\(1.6\)-approximation algorithm of \citet{Blauth2026} gives a better-than-\(3.2\) PCTSP-reduction algorithm.

Our main contribution is a polynomial time approximation algorithm that treats edge profits directly, rather than through the above edge-to-vertex reduction. Our work is highly motivated by the recent PCTSP algorithm of \citet{Blauth2026}. We first transform the input instance into a preprocessed complete graph. We show that some optimal solution to the original PCRPP instance corresponds to a canonical solution to the PCRPP instance on this preprocessed complete graph with objective value no larger than the optimal value of the original instance. We then formulate an integer programming model for these canonical solutions and obtain a linear programming relaxation. The optimal value of the resulting linear programming relaxation is used as a lower bound in the approximation analysis. We solve this linear programming relaxation and, based on an optimal solution to this linear programming relaxation, use splitting-off techniques, our edge-profit tree decomposition, and a pruning step to construct a polynomial-size set of trees. Although splitting-off is length-nonincreasing in a complete metric graph, the preprocessed complete graph used here is not a complete metric graph; a technical part of the analysis proves that the splitting-off operations used by the algorithm still do not increase the length term of the linear programming relaxation. The key point of our edge-profit tree decomposition is that it preserves the coupling between each positive-profit edge and its two endpoints, so that edge profits can be analyzed directly. For each tree in this set, the algorithm expands the relevant edges in the preprocessed complete graph back to paths in the original graph and then performs parity correction to obtain an Eulerian multigraph. An Eulerian circuit in this multigraph gives a candidate solution, and the algorithm returns the candidate solution with minimum objective value. This gives a polynomial time approximation algorithm with an approximation ratio strictly smaller than \(1.6\) for the minimization version of the PCRPP.

This guarantee matches the best known better-than-\(1.6\)-approximation guarantee for the PCTSP. This comparison is meaningful because the PCTSP admits a direct approximation-preserving reduction to the PCRPP. Given a PCTSP instance, we keep all edges in the PCTSP instance as zero-profit edges and, for every non-root vertex \(v\), add a zero-length leaf edge \(vv'\) with profit equal to the penalty of \(v\). A rooted closed walk can collect this edge profit at zero additional length exactly when it visits \(v\). Under the natural correspondence between rooted closed walks, the objective value of a walk in the constructed PCRPP instance is exactly the objective value of the corresponding walk in the original PCTSP instance.

We also conduct computational experiments on the public 118-instance benchmark set of \citet{Araoz2009}. Since the original benchmark reports values for the maximization formulation, we convert the reported optimal values to the minimization objective used in this paper. The experiments compare our algorithm with these converted optimal values and with the PCTSP-reduction algorithm in Appendix~\ref{app:pctsp-reduction}, which also has a worst-case approximation guarantee. Thus the experimental comparison focuses on algorithms with approximation guarantees.

The remainder of the paper is organized as follows. Section~\ref{sec:prelim} introduces preliminaries. Section~\ref{sec:approx_alg} presents the algorithm for the PCRPP. Section~\ref{sec:analysis} analyzes the approximation ratio of the algorithm. Section~\ref{sec:experiments} reports the computational experiments, and Section~\ref{sec:conclusion} concludes the paper.

\section{Preliminaries}
\label{sec:prelim}

\subsection{Notations}
\label{sec:notation}
Let \(G=(V,E;w,p;r)\) be a connected PCRPP instance, where \(r\in V\) is the root vertex, \(w_e\ge0\) is the length of edge \(e\), and \(p_e\ge0\) is its profit. For a subset of vertices \(S\subseteq V\), we denote by \(\delta(S)\) the set of edges with exactly one endpoint in \(S\), and by \(E(S)\) the set of edges with both endpoints in \(S\). For two disjoint vertex sets \(S_1,S_2\subseteq V\), let \(\delta(S_1,S_2)\) denote the set of edges with one endpoint in \(S_1\) and the other endpoint in \(S_2\). For a single vertex \(v\in V\), we write \(\delta(v)\) instead of \(\delta(\{v\})\). Given an edge vector \(x\in\mathbb R^E\), an edge set \(F\subseteq E\), and a vertex set \(S\subseteq V\), define \(x(F)=\sum_{e\in F}x_e\), \(w(F)=\sum_{e\in F}w_e\), and let \(V(F)\) be the set of vertices incident to at least one edge in \(F\). We use \(\chi_E^F\in\{0,1\}^E\) and \(\chi_V^S\in\{0,1\}^V\) to denote the edge-incidence vector of \(F\) and vertex-incidence vector of \(S\), respectively; that is, \((\chi_E^F)_e=1\) if and only if \(e\in F\), and \((\chi_V^S)_v=1\) if and only if \(v\in S\). The set of positive-profit edges is denoted by \(E^+=\{e\in E:p_e>0\}\), and the set of zero-profit edges is denoted by \(E^0=E\setminus E^+\). When edge multisets are used, their multiset union is denoted by \(\uplus\). For a graph or multigraph \(H\), we write \(w(H)\) for the total length of all edges in \(H\), counting an edge as many times as it appears, and \(\operatorname{odd}(H)\) for the set of vertices of odd degree in \(H\). For a vertex set \(S\subseteq V(H)\), let \(\delta_H(S)\) denote the set, or multiset, of edges of \(H\) with exactly one endpoint in \(S\), with multiplicity counted when \(H\) is a multigraph. By the handshaking lemma, \(\operatorname{odd}(H)\) has even cardinality. For a walk \(W\), we write \(w(W)\) for its length, counting each edge as many times as it is traversed. We say that an edge \(e=uv\) satisfies the triangle inequality if its length \(w_e\) is at most the total length of every \(u\)-to-\(v\) path in \(G\). If all edges in a graph satisfy the triangle inequality, we call the graph a metric graph.

\subsection{Fundamental lemmas}
\label{sec:lemmas}
We recall several standard tools used in our PCRPP algorithm and its analysis. We first recall the fractional splitting-off technique of \citet{Lovasz1976}, \citet{Mader1978}, and \citet{Frank1992}. Let \(G=(V,E)\) be a complete graph and let \(x\in\mathbb R_{\ge0}^E\) be a nonnegative edge vector. For two vertices \(s,t\in V\), the minimum \(s\)-\(t\) cut size with respect to \(x\) is $\min\{x(\delta(S)): S\subseteq V,\ s\in S,\ t\notin S\}.$ A splitting-off operation at a vertex \(v\in V\) is defined as follows. In the non-degenerate case, choose two distinct vertices \(u,w\in V\setminus\{v\}\) and a number \(0<\varepsilon\le \min\{x_{vu},x_{vw}\}\); then decrease \(x_{vu}\) and \(x_{vw}\) by \(\varepsilon\), and increase \(x_{uw}\) by \(\varepsilon\). In the degenerate case, choose a vertex \(u\in V\setminus \{v\}\) and a number \(0<\varepsilon\le x_{vu}\), and simply decrease \(x_{vu}\) by \(\varepsilon\). Such a splitting-off operation is feasible if, after the operation, the minimum \(s\)-\(t\) cut size does not decrease for any two vertices \(s,t\in V\setminus\{v\}\). A complete splitting at \(v\) is a sequence of feasible splitting-off operations after which \(x_e=0\) for every edge \(e\in\delta(v)\).

\begin{lemma}[{\citealp[Theorem~4.1]{Blauth2023}}]
\label{lem:splitting}
Let \(G=(V,E;w)\) be a complete graph, let \(x\in\mathbb R_{\ge0}^E\), and let \(v\in V\). A complete splitting at \(v\) can be computed in polynomial time. That is, one can compute a sequence of feasible splitting-off operations at \(v\) and obtain an edge vector \(x'\) such that \(x'_e=0\) for every \(e\in\delta(v)\), while the minimum \(s\)-\(t\) cut size does not decrease for any two vertices \(s,t\in V\setminus\{v\}\). Moreover, if \(G=(V,E;w)\) is a metric complete graph, then
\[
\sum_{e\in E} w_e x'_e\le \sum_{e\in E} w_e x_e .
\]
\end{lemma}

We next recall an exact tree decomposition that constructs a probability distribution over trees from a feasible solution to a linear programming relaxation for the PCTSP. We first state the PCTSP linear programming relaxation. On a metric complete graph \(G=(V,E)\), with root vertex \(r\), edge lengths \(w_e\), and vertex penalties \(\pi_v\), the PCTSP linear programming relaxation, denoted by \textsc{PCTSP-LP}, is

\begin{align}
\min \quad &
\sum_{e\in E} w_e x_e + \sum_{v\in V}\pi_v(1-y_v)
\label{eq:PCTSP-lp-obj}\\
\text{\textit{s.t.}} \quad
& x(\delta(v)) = 2y_v,
&& \forall\, v\in V\setminus\{r\}, \label{eq:PCTSP-lp-deg1}\\
& x(\delta(r)) \le 2, &&
\label{eq:PCTSP-lp-deg2}\\
& x(\delta(S)) \ge 2y_v,
&& \forall\, S\subseteq V\setminus\{r\},\ \forall\, v\in S,
\label{eq:PCTSP-lp-cut}\\
& y_r = 1, &&
\label{eq:PCTSP-lp-root}\\
& x_e \ge 0,
&& \forall\, e\in E, \label{eq:PCTSP-lp-x}\\
& 0\le y_v \le 1,
&& \forall\, v\in V\setminus\{r\}. \label{eq:PCTSP-lp-y}
\end{align}

The following exact tree decomposition was obtained by \citet{Blauth2026} for feasible solutions to PCTSP-LP. Although the PCTSP is usually defined on a metric complete graph, this decomposition is purely structural: it uses only feasibility of the linear programming relaxation and does not rely on the metric property of the edge lengths. Its proof applies complete splittings to a solution to PCTSP-LP and then undoes the resulting splitting-off operations in reverse order to construct the trees. We use the lemma only as a structural decomposition result.

\begin{lemma}[ {\citealp[Lemma~5]{Blauth2026}}]
\label{lem:treedecomp}
Let \((x,y)\) be a feasible solution to \eqref{eq:PCTSP-lp-obj}--\eqref{eq:PCTSP-lp-y}. Suppose that there is an edge \(e_0=rv_0\in E\) such that \(x_{e_0}\ge1\) and \(y_{v_0}=1\). Then one can compute in polynomial time a polynomial-size set \(\mathcal T\) of trees and a probability vector \(\lambda\in[0,1]^{\mathcal T}\) where \(\lambda_T\) is the probability assigned to \(T\), such that \(\sum_{T\in\mathcal T}\lambda_T=1\), every \(T\in\mathcal T\) contains the root vertex \(r\), and
\[
\sum_{T\in\mathcal T}\lambda_T\chi_E^{E(T)}
=
x-\chi_E^{\{e_0\}},
\qquad
\sum_{T\in\mathcal T}\lambda_T\chi_V^{V(T)}
=
y .
\]
Equivalently, if \(T\) is sampled according to \(\lambda\), then
\[
\Pr[e\in E(T)]
=
\begin{cases}
x_{e}-1, &  e=e_0,\\
x_e, & \forall e\in E\setminus\{e_0\},
\end{cases} 
\qquad
\Pr[v\in V(T)]=y_v
\quad \forall\, v\in V .
\]
\end{lemma}

Finally, we recall the parity-correction tool used after a tree has been constructed. Let \(G=(V,E;w)\) be a weighted graph and let \(Q\subseteq V\) have even cardinality. A \(Q\)-join is an edge multiset \(J\) over \(E\) such that a vertex \(v\) has odd degree in \((V,J)\) if and only if \(v\in Q\). For a graph or multigraph \(H\), adding an \(\operatorname{odd}(H)\)-join to \(H\), by multiset union, makes all degrees even.

The linear programming relaxation for the minimum-length \(Q\)-join, denoted by \textsc{Q-Join-LP}, is
\begin{align}
\min \quad &
\sum_{e\in E} w_e x_e
\label{eq:qjoin-lp-obj}\\
\text{\textit{s.t.}} \quad
& x(\delta(S))\ge 1,
&& \forall\, S\subseteq V \text{ with } |S\cap Q| \text{ odd},
\label{eq:qjoin-lp-cut}\\
& x_e\ge0,
&& \forall\, e\in E.
\label{eq:qjoin-lp-nonneg}
\end{align}

\begin{lemma}[{\citealp{Edmonds1973}}]
\label{lem:tjoin}
Let \(G=(V,E;w)\) be a weighted graph, and let \(Q\subseteq V\) have even cardinality. A minimum-length \(Q\)-join can be computed in polynomial time, and its length is no more than the length of any feasible fractional solution to \textup{\textsc{Q-Join-LP}}.
\end{lemma}

\section{Approximation algorithms}
\label{sec:approx_alg}

In this section, we present our approximation algorithm for the PCRPP. The algorithm treats edge profits directly, rather than using the PCTSP-reduction approach described in the introduction. We first preprocess the input instance into a complete graph and show that some optimal solution to the original PCRPP instance corresponds to a canonical solution to the PCRPP instance on this complete graph with objective value no larger than the optimal value of the original instance. We then formulate an integer programming model for these canonical solutions and solve its linear programming relaxation. Starting from an optimal solution to this relaxation, we use splitting-off operations and an edge-profit tree decomposition to construct a polynomial-size set of rooted trees. The key feature of this decomposition is that it preserves the coupling between each positive-profit edge and its two endpoints. The pruning step is incorporated into the construction of the candidate trees. For each candidate tree, we restore the relevant complete-graph edges to paths in the original graph, obtaining a connected edge multiset that contains the root. We then add a minimum-length \(Q\)-join to correct parities, obtain an Eulerian multigraph, extract a rooted closed walk as a candidate solution, and output the candidate solution with minimum length plus uncollected profits.

The remainder of this section formalizes the preprocessing step, the integer programming formulation for canonical solutions and its linear programming relaxation, and the best-of-many algorithm for the PCRPP.

\subsection{Preprocessing}
\label{sec:preprocessing}

Let \(G=(V,E;w,p;r)\) be the original connected PCRPP instance. We transform \(G\) into a preprocessed complete graph \(\hat G=(\hat V,\hat E;\hat w,\hat p;r)\). The purpose of the preprocessing is to separate positive-profit edges while keeping their lengths and profits unchanged, and then to complete the graph by zero-profit shortest path edges.

First, we modify the root. Let \(k\) be the number of positive-profit edges incident to \(r\). If \(k>0\), introduce \(k\) new vertices \(r_1,\ldots,r_k\). Reassign the endpoint \(r\) of each positive-profit edge incident to \(r\) to a distinct vertex \(r_i\), and add a zero-length, zero-profit edge \(rr_i\) for each \(i=1,2,\dots,k\).

Second, we modify the remaining vertices. For each vertex \(v\in V\setminus\{r\}\), let \(k\) be the number of positive-profit edges incident to \(v\) after the root modification. If \(k>1\), introduce \(k\) new vertices \(v_1,\ldots,v_k\). Reassign the endpoint \(v\) of each positive-profit edge incident to \(v\) to a distinct vertex \(v_i\), and add a zero-length, zero-profit edge \(vv_i\) for each \(i=1,2,\dots,k\).

We refer to the two steps above as the vertex-copying step. Let \(\hat V\) be the vertex set obtained after this step. We construct the complete graph on \(\hat V\) as follows. Every positive-profit edge obtained from an original positive-profit edge is kept with its length and profit. For every other pair \(u,v\in\hat V\), we add the edge \(uv\), set \(\hat p_{uv}=0\), and define \(\hat w_{uv}\) as the shortest path length between \(u\) and \(v\) after the vertex-copying step. We denote the set of positive-profit edges of \(\hat G\) by \(\hat E^+\).

The preprocessed complete graph has the following properties.
\begin{enumerate}
 \item[(P1)] The root vertex \(r\) is not incident to any positive-profit edge.
 \item[(P2)] Every vertex in \(\hat V\) is incident to at most one positive-profit edge. Hence the edges in \(\hat E^+\) are pairwise vertex-disjoint.
 \item[(P3)] Every zero-profit edge \(uv\in \hat E\setminus \hat E^+\) satisfies the triangle inequality: its length \(\hat w_{uv}\) is at most the total length of any \(u\)-to-\(v\) path in \(\hat G\).
\end{enumerate}
Properties (P1) and (P2) follow directly from the vertex-copying step. Property (P3) is immediate from the shortest path definition of the zero-profit edges.

After preprocessing, we formulate an integer programming model on the preprocessed complete graph \(\hat G=(\hat V,\hat E;\hat w,\hat p;r)\). This model describes the canonical solutions that will be used in the lower-bound argument. Here \(\hat w_e\) and \(\hat p_e\) denote the length and profit of edge \(e\in\hat E\), respectively, and \(\hat E^+=\{e\in\hat E:\hat p_e>0\}\). The model is
\begin{align}
\min \quad &
\sum_{e\in \hat{E}} \hat{w}_e x_e
+
\sum_{e\in \hat{E}} \hat{p}_e (1-x_e)
\label{eq:pcrpp-ip-obj}\\
\text{\textit{s.t.}} \quad
& x(\delta(v)) = 2y_v,
&& \forall\, v\in \hat{V}\setminus\{r\},
\label{eq:pcrpp-ip-deg1}\\
& x(\delta(r)) \le 2,
&&
\label{eq:pcrpp-ip-deg2}\\
& x(\delta(S)) \ge 2y_v,
&& \forall\, S\subseteq \hat{V}\setminus\{r\},\ \forall\, v\in S,
\label{eq:pcrpp-ip-cut}\\
& y_u = y_v = x_{uv},
&& \forall\, uv\in \hat{E}^+,
\label{eq:pcrpp-ip-couple}\\
& y_r = 1,
&&
\label{eq:pcrpp-ip-root}\\
& x_e \in \{0,1\},
&& \forall\, e\in \hat{E},
\label{eq:pcrpp-ip-x}\\
& y_v \in \{0,1\},
&& \forall\, v\in \hat{V}\setminus\{r\}.
\label{eq:pcrpp-ip-y}
\end{align}
We denote this integer programming model by \textsc{PCRPP-IP}. The model is not intended to describe every feasible solution to the PCRPP instance on \(\hat G\). Rather, it describes a restricted class of canonical solutions that is sufficient for the lower-bound proof: as shown in Theorem~\ref{thm:lp_bound}, some optimal solution to the original PCRPP instance corresponds to a feasible solution to \textsc{PCRPP-IP} with objective value no larger than the optimal value of the original instance.

The linear programming relaxation of \textsc{PCRPP-IP}, denoted by \textsc{PCRPP-LP}, is obtained by replacing \eqref{eq:pcrpp-ip-x} and \eqref{eq:pcrpp-ip-y} with
\[
0\le x_e\le 1 \quad \forall\, e\in\hat E^+,\qquad
x_e\ge0 \quad \forall\, e\in\hat E\setminus\hat E^+,\qquad
0\le y_v\le1 \quad \forall\, v\in\hat V\setminus\{r\}.
\]
The upper bound on \(x_e\) is imposed only for positive-profit edges; for zero-profit edges, nonnegativity is sufficient in the relaxation. Although \textsc{PCRPP-LP} contains exponentially many cut constraints, it is solvable in polynomial time by the ellipsoid method \citep{Khachiyan1979,Grotschel1981}. Indeed, all constraints except \eqref{eq:pcrpp-ip-cut} can be checked directly, and the cut constraints can be separated by a polynomial number of maximum-flow or minimum-cut computations.

We next state two theorems about the linear programming relaxation \textsc{PCRPP-LP}; both theorems are proved in Section~\ref{sec:analysis}. Theorem~\ref{thm:lp_bound} states that the optimal objective value of \textsc{PCRPP-LP} on the preprocessed complete graph gives a valid lower bound for the optimal objective value of the original PCRPP instance.

\begin{theorem}
\label{thm:lp_bound}
Let \(\mathrm{OPT}\) be the optimal objective value of the PCRPP on the original graph \(G\), and let \(\mathrm{OPT}_{\mathrm{LP}}\) be the optimal objective value of \textup{\textsc{PCRPP-LP}} on the preprocessed complete graph \(\hat G\). Then
\[
\mathrm{OPT}_{\mathrm{LP}}\le \mathrm{OPT}.
\]
\end{theorem}

Theorem~\ref{thm:pcrpp_treedecomp} gives the edge-profit tree decomposition used by our algorithm. Starting from any feasible solution to \textsc{PCRPP-LP}, it constructs a probability distribution over rooted trees in the preprocessed complete graph. The edge-profit tree decomposition controls the expected length of the tree, preserves the probability with which each positive-profit edge appears in the tree, and, crucially, every tree in the support preserves the coupling between each positive-profit edge and its two endpoints.

\begin{theorem}
\label{thm:pcrpp_treedecomp}
Given a feasible solution \((x,y)\) to \textup{\textsc{PCRPP-LP}} on \(\hat G\), one can compute in polynomial time a polynomial-size set \(\mathcal T\) of trees in \(\hat G\) and a probability vector \(\lambda\in[0,1]^{\mathcal T}\) where \(\lambda_T\) is the probability assigned to \(T\), such that \(\sum_{T\in\mathcal T}\lambda_T=1\) and every \(T\in\mathcal T\) contains the root vertex \(r\). If a tree \(T\) is sampled according to \(\lambda\), then
\[
\mathbb E\left[\hat w(T)\right]
\le
\sum_{e\in\hat E}\hat w_e x_e.
\]
Moreover, for every vertex \(v\in\hat V\setminus\{r\}\), and for every positive-profit edge \(e\in\hat E^+\),
\[
\Pr[v\in V(T)]=y_v,\qquad \Pr[e\in E(T)]=x_e.
\]
Furthermore, every tree \(T\in\mathcal T\) with \(\lambda_T>0\) satisfies the following coupling property: for every positive-profit edge \(e=uv\in\hat E^+\),
\[
e\in E(T)
\quad\Longleftrightarrow\quad
u\in V(T)
\quad\Longleftrightarrow\quad
v\in V(T).
\]
\end{theorem}

\subsection{The best-of-many algorithm}
\label{sec:algorithm}

Before presenting the algorithm, we define the edge-profit core used in the pruning step. Let \(T\) be a tree in \(\hat G\) that contains the root vertex \(r\), let \(x\in\mathbb R_{\ge0}^{\hat E}\) be an edge vector, and let \(\gamma\in[0,1]\). The edge-profit core of \(T\) with respect to \(x\) and \(\gamma\), denoted by \(\operatorname{core}_{x}(T,\gamma)\), is the minimal subtree of \(T\) that contains \(r\) and contains every positive-profit edge \(e\in E(T)\cap\hat E^+\) with \(x_e\ge\gamma\). Equivalently, if \(T\) contains at least one positive-profit edge \(e\) with \(x_e\ge\gamma\), then \(\operatorname{core}_{x}(T,\gamma)\) is the union of the unique paths in \(T\) from \(r\) to the endpoints of all positive-profit edges \(e\in E(T)\cap\hat E^+\) with \(x_e\ge\gamma\). If no such edge exists, we set \(\operatorname{core}_{x}(T,\gamma)=(\{r\},\emptyset)\). The edge-profit core can be constructed by the following pruning procedure: repeatedly delete a leaf \(v\neq r\) and its incident edge whenever \(v\) is not incident to a positive-profit edge \(e\in E(T)\cap\hat E^+\) with \(x_e\ge\gamma\). The procedure stops when every leaf distinct from \(r\) is incident to such a positive-profit edge \(e\in E(T)\cap\hat E^+\) with \(x_e\ge\gamma\). Figure~\ref{fig:coretree} illustrates this edge-profit core.

\begin{figure}[htbp]\small
 \centering
 \includegraphics[width=0.7\textwidth]{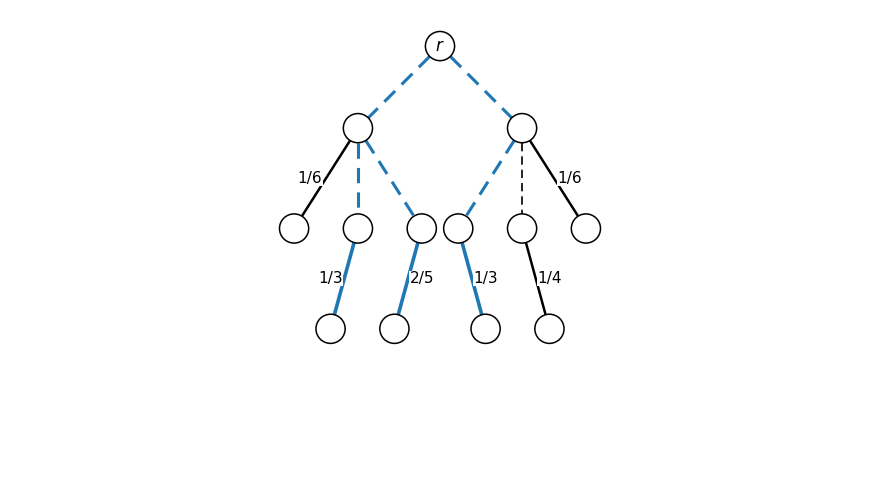}
 \caption{An illustration of the edge-profit core. Solid edges are positive-profit edges, dashed edges are zero-profit edges, and blue edges are contained in \(\operatorname{core}_{x}(T,\gamma)\) for \(\gamma=1/3\). The numbers on positive-profit edges indicate the corresponding values of \(x_e\).}
 \label{fig:coretree}
\end{figure}

With the edge-profit core defined, we now describe the best-of-many algorithm. For each threshold \(\delta\), the algorithm applies complete splitting-off operations, using Lemma~\ref{lem:splitting} as the splitting tool. Lemma~\ref{lem:low_y_splitting_pcrpp} guarantees that the resulting pair \((\tilde x,\tilde y)\) is feasible for \textsc{PCRPP-LP}. Hence the edge-profit tree decomposition of Theorem~\ref{thm:pcrpp_treedecomp} can be applied to \((\tilde x,\tilde y)\). For each tree produced by this decomposition and each relevant threshold \(\gamma\), the algorithm forms the corresponding edge-profit core, restores it to the original graph, adds a minimum-length \(Q\)-join to correct parities, and evaluates the resulting rooted closed walk as a candidate solution. The algorithm returns the best candidate solution. We summarize the procedure in Algorithm~\ref{alg:best_of_many}.

\begin{algorithm}[htbp]\small
\caption{Best-of-Many Algorithm for the PCRPP}
\label{alg:best_of_many}
\begin{algorithmic}[1]
\State \textbf{Input:} A connected PCRPP instance \(G=(V,E;w,p;r)\).
\State Apply the preprocessing of Section~\ref{sec:preprocessing} to obtain the preprocessed complete graph \(\hat G=(\hat V,\hat E;\hat w,\hat p;r)\) and the set of positive-profit edges \(\hat E^+\).
\State Solve \textsc{PCRPP-LP} on \(\hat G\), and let \((x^*,y^*)\) be an optimal solution.
\State Initialize \(W_A\) as the trivial rooted closed walk that stays at \(r\), and set \(\mathrm{ALG}\) to its objective value.
\For{each outer threshold \(\delta\in\{y_v^*:v\in\hat V\setminus\{r\}\}\)}
 \State Set \((\tilde x,\tilde y)\leftarrow(x^*,y^*)\).
 \For{each \(v\in\hat V\setminus\{r\}\) with \(0<\tilde y_v<\delta\)}
 \State Apply a complete splitting at \(v\), using Lemma~\ref{lem:splitting}, update the current edge vector \(\tilde x\), and set \(\tilde y_v\leftarrow0\).
 \EndFor
 \State By Lemma~\ref{lem:low_y_splitting_pcrpp}, the resulting pair \((\tilde x,\tilde y)\) is feasible for \textsc{PCRPP-LP}. Apply Theorem~\ref{thm:pcrpp_treedecomp} to \((\tilde x,\tilde y)\), obtaining a polynomial-size set \(\mathcal T\) of trees.
 \For{each \(T\in\mathcal T\)}
 \For{each inner threshold \(\gamma\in \{\tilde x_e:e\in E(T)\cap\hat E^+\}\)}
 \State Construct the edge-profit core \(\operatorname{core}_{\tilde x}(T,\gamma)\) in \(\hat G\).
 \State Restore \(\operatorname{core}_{\tilde x}(T,\gamma)\) to the original graph by replacing each zero-profit edge by the shortest path defining its length, and then merging the copied vertices back to their original vertices. Let \(H_{T,\gamma}\) be the resulting connected multigraph in \(G\).
 \State Let \(\operatorname{odd}(H_{T,\gamma})\) denote the set of vertices of odd degree in \(H_{T,\gamma}\), compute a minimum-length \(\operatorname{odd}(H_{T,\gamma})\)-join \(J_{T,\gamma}\) in the original graph \(G\), using Lemma~\ref{lem:tjoin}.
 \State Form the Eulerian multigraph \(M_{T,\gamma}=H_{T,\gamma}\uplus J_{T,\gamma}\).
 \State Extract an Eulerian tour of \(M_{T,\gamma}\), viewed as a closed walk \(W_{T,\gamma}\) rooted at \(r\) in \(G\).
 \State Let \(A_{T,\gamma}\) be the length of \(W_{T,\gamma}\) plus the total profit of the edges not traversed by \(W_{T,\gamma}\).
 \If{\(A_{T,\gamma}<\mathrm{ALG}\)}
 \State \(\mathrm{ALG}\leftarrow A_{T,\gamma}\).
 \State \(W_A\leftarrow W_{T,\gamma}\).
 \EndIf
 \EndFor
 \EndFor
\EndFor
\State \textbf{Output:} \(W_A\) and its objective value \(\mathrm{ALG}\).
\end{algorithmic}
\end{algorithm}

The algorithm always outputs a feasible solution to the PCRPP instance and runs in polynomial time. The preprocessing step runs in polynomial time: it creates at most two copied vertices for each edge of the original graph, computes the required shortest path lengths, and constructs the preprocessed complete graph \(\hat G\). Moreover, \(|\hat V|\le |V|+2|E|\) and \(|\hat E|=O(|\hat V|^2)=O((|V|+|E|)^2)\), so the preprocessed instance has polynomial size. The linear programming relaxation \(\textsc{PCRPP-LP}\) can be solved in polynomial time. For each outer threshold \(\delta\in\{y_v^*:v\in\hat V\setminus\{r\}\}\), the complete splittings can be computed in polynomial time. The number of outer thresholds is at most \(|\hat V|\). By Theorem~\ref{thm:pcrpp_treedecomp}, the set \(\mathcal T\) has polynomial size and can be computed in polynomial time. For each tree \(T\), the number of inner thresholds \(\gamma\in\{\tilde x_e:e\in E(T)\cap\hat E^+\}\) is at most \(|\hat E^+|\le |E|\). For every enumerated pair \((T,\gamma)\), constructing the edge-profit core, restoring it to the original graph, computing a minimum-length \(Q\)-join, and evaluating the resulting rooted closed walk can all be done in polynomial time. Therefore the algorithm generates and evaluates only polynomially many candidate solutions, each in polynomial time.

\section{Theoretical Analysis}
\label{sec:analysis}

In this section, we prove the statements used in the analysis of the approximation algorithm. We first prove Theorem~\ref{thm:lp_bound}, which shows that the optimal value of the linear programming relaxation \textsc{PCRPP-LP} on the preprocessed complete graph \(\hat G\) gives a valid lower bound for the optimal value of the original PCRPP instance.

\begin{proof}[Proof of Theorem~\ref{thm:lp_bound}]
Let \(W^*\) be an optimal rooted closed walk in the original graph \(G\), and let \(E^+(W^*)\subseteq E^+\) be the set of positive-profit edges traversed by \(W^*\). If \(E^+(W^*)=\emptyset\), then the trivial canonical solution to the PCRPP instance on the preprocessed complete graph, with \(x=0\), \(y_r=1\), and \(y_v=0\) for all \(v\in\hat V\setminus\{r\}\), is feasible for \textsc{PCRPP-IP} and has objective value \(\sum_{e\in E^+}p_e\le \mathrm{OPT}\). Hence it is also a feasible solution to \textsc{PCRPP-LP} with value at most \(\mathrm{OPT}\).

Assume now that \(E^+(W^*)\neq\emptyset\). Let \(\hat E^+(W^*)\subseteq\hat E^+\) be the set of positive-profit edges in \(\hat G\) corresponding to the edges in \(E^+(W^*)\). These edges have the same lengths and profits as their corresponding original edges, are pairwise vertex-disjoint by Property~(P2), and are not incident to \(r\) by Property~(P1).

Traverse \(W^*\) from \(r\), and order the edges in \(E^+(W^*)\) by the time at which they are first traversed. For the \(i\)-th edge in this order, where \(i=1,\ldots,|E^+(W^*)|\), let \(u_i^*v_i^*\in\hat E^+(W^*)\) be the corresponding positive-profit edge in \(\hat G\). The notation is chosen so that \(u_i^*\) corresponds to the endpoint of the original edge first visited by \(W^*\) when this edge is first traversed, and \(v_i^*\) corresponds to the endpoint visited immediately afterwards. Define the canonical solution \(\hat W^*\) in \(\hat G\) by
\[
\hat W^* = r\to u_1^*\to v_1^*\to u_2^*\to v_2^*
\to \cdots \to
u_{|E^+(W^*)|}^*\to v_{|E^+(W^*)|}^*\to r .
\]
In this sequence, each edge \(u_i^*v_i^*\) is a positive-profit edge in \(\hat E^+(W^*)\), and all other edges are zero-profit edges of \(\hat G\). Since the vertices \(u_i^*,v_i^*\) are all distinct and none of them is \(r\), \(\hat W^*\) is a simple cycle containing \(r\).

Let \(\hat x\) be the edge-incidence vector of \(\hat W^*\), and let \(\hat y_v=1\) exactly for the vertices appearing in \(\hat W^*\). Then \((\hat x,\hat y)\) is feasible for \textsc{PCRPP-IP}. The constraints \eqref{eq:pcrpp-ip-deg1}--\eqref{eq:pcrpp-ip-cut} and \eqref{eq:pcrpp-ip-root}--\eqref{eq:pcrpp-ip-y} follow from the fact that \(\hat W^*\) is a simple cycle containing \(r\). For the coupling constraints \eqref{eq:pcrpp-ip-couple}, let \(e=uv\in\hat E^+\). If \(e\in\hat E^+(W^*)\), then \(e=u_i^*v_i^*\) for some \(i\), and hence \(\hat x_e=\hat y_u=\hat y_v=1\). If \(e\notin\hat E^+(W^*)\), then neither endpoint of \(e\) appears in \(\hat W^*\), by Properties~(P1) and (P2), and therefore \(\hat x_e=\hat y_u=\hat y_v=0\). Thus \((\hat x,\hat y)\) is also feasible for \textsc{PCRPP-LP}.

It remains to compare objective values. The positive-profit edges of \(\hat W^*\) have the same total length as the corresponding edges in \(E^+(W^*)\). By Property~(P3), the total length of the zero-profit edges that connect the positive-profit edges in \(\hat W^*\) is at most the total length of the corresponding paths in \(W^*\) that connect the edges in \(E^+(W^*)\). Therefore the length of \(\hat W^*\) is at most the length of \(W^*\). Moreover, \(\hat W^*\) collects exactly the positive-profit edges in \(\hat E^+(W^*)\), so its uncollected-profit term is exactly the total profit of the positive-profit edges not traversed by \(W^*\). Hence the objective value of \((\hat x,\hat y)\) is at most the objective value of \(W^*\), which is \(\mathrm{OPT}\). Consequently, \(\mathrm{OPT}_{\mathrm{LP}}\le \mathrm{OPT}\).
\end{proof}

We now prove Theorem~\ref{thm:pcrpp_treedecomp} about the edge-profit tree decomposition.

\begin{proof}[Proof of Theorem~\ref{thm:pcrpp_treedecomp}]
Let \((x,y)\) be a feasible solution to \textsc{PCRPP-LP} on \(\hat G=(\hat V,\hat E;\hat w,\hat p;r)\). Construct an auxiliary complete graph \(G^A=(V^A,E^A)\) by adding a copy \(r'\) of the root vertex \(r\), and set \(V^A=\hat V\cup\{r'\}\). The root remains \(r\). The auxiliary edge \(e_0=rr'\) has length zero; for every \(v\in\hat V\setminus\{r\}\), the edge \(r'v\) has the same length as \(rv\); and all edges of \(\hat G\) keep their original lengths.

Define an auxiliary edge vector \(\bar x\) and vertex vector \(\bar y\) on \(G^A\) as follows. Set \(\bar y_v=y_v\) for \(v\in\hat V\) and \(\bar y_{r'}=1\). For the edge variables, set \(\bar x_{uv}=x_{uv}\) for all \(u,v\in\hat V\setminus\{r\}\), \(\bar x_{rv}=\bar x_{r'v}=x_{rv}/2\) for all \(v\in\hat V\setminus\{r\}\), and \(\bar x_{e_0}=2-x(\delta(r))/2\). Since \(x(\delta(r))\le2\), we have \(\bar x_{e_0}\ge1\).

We verify that \((\bar x,\bar y)\) satisfies \eqref{eq:PCTSP-lp-deg1}--\eqref{eq:PCTSP-lp-y} on \(G^A\). For every \(v\in\hat V\setminus\{r\}\), the definition of \(\bar x\) gives \(\bar x(\delta(v))=x(\delta(v))=2y_v=2\bar y_v\). Also, \(\bar x(\delta(r'))=\bar x_{e_0}+\sum_{v\in\hat V\setminus\{r\}}\bar x_{r'v} =2=2\bar y_{r'}\). Hence the degree constraints \eqref{eq:PCTSP-lp-deg1} hold for all vertices in \(V^A\setminus\{r\}\). The root satisfies \(\bar x(\delta(r))=2\), so \eqref{eq:PCTSP-lp-deg2} holds. The constraint \eqref{eq:PCTSP-lp-root} holds because \(\bar y_r=y_r=1\), and the nonnegativity constraints \eqref{eq:PCTSP-lp-x} and \eqref{eq:PCTSP-lp-y} follow from the definition of \((\bar x,\bar y)\) and the feasibility of \((x,y)\) for \textsc{PCRPP-LP}. We next verify the cut constraints \eqref{eq:PCTSP-lp-cut}. Fix \(v\in V^A\setminus\{r\}\) and \(S\subseteq V^A\setminus\{r\}\) with \(v\in S\). First suppose that \(r'\in S\). By the definition of \(\bar x\),
\[
\begin{aligned}
\bar x(\delta(S))
&\ge
\bar x_{e_0}
+\bar x\bigl(\delta(\{r\},S\setminus\{r'\})\bigr)
+\bar x\bigl(\delta(V^A\setminus(S\cup\{r\}), \{r'\})\bigr) \\
&=\bar x_{e_0}
+\bar x\bigl(\delta(\{r\},S\setminus\{r'\})\bigr)
+\bar x\bigl(\delta(\hat V\setminus(S\cup\{r\}), \{r'\})\bigr)
\\
&=
\bar x_{e_0}
+\frac12 x\bigl(\delta(\{r\},S\setminus\{r'\})\bigr)
+\frac12 x\bigl(\delta(\{r\},\hat V\setminus(S\cup\{r\}))\bigr) \\
&=
\bar x_{e_0}+\frac12 x(\delta(r))
=
2
\ge
2\bar y_{v} .
\end{aligned}
\]
Now suppose that \(r'\notin S\). Then \(S\subseteq V^A\setminus\{r,r'\}=\hat V\setminus\{r\}\), and the definition of \(\bar x\) gives \(\bar x(\delta(S))=x(\delta(S))\). Hence \(\bar x(\delta(S))\ge2y_v=2\bar y_v\) by the cut constraint \eqref{eq:pcrpp-ip-cut} of \textsc{PCRPP-LP}. Hence \eqref{eq:PCTSP-lp-cut} holds.

Thus \((\bar x,\bar y)\) is feasible for \textsc{PCTSP-LP} on \(G^A\). The distinguished edge required by Lemma~\ref{lem:treedecomp} is \(e_0=rr'\); indeed, \(\bar x_{e_0}\ge1\) and \(\bar y_{r'}=1\). Applying Lemma~\ref{lem:treedecomp}, we obtain a polynomial-size set \(\mathcal L^A\) of rooted trees in \(G^A\) and a probability vector \(\lambda^A\in[0,1]^{\mathcal L^A}\), where \(\lambda^A_L\) is the probability assigned to \(L\). For every \(e\in E^A\setminus\{e_0\}\), we have \(\Pr[e\in E(L)]=\bar x_e\), while \(\Pr[e_0\in E(L)]=\bar x_{e_0}-1\). We will use \(\Pr[v\in V(L)]=\bar y_v\) to prove the coupling between each positive-profit edge and its two endpoints.

For each \(L\in\mathcal L^A\), obtain \(\Phi(L)\) by merging \(r\) and \(r'\) and, if necessary, deleting one edge from the resulting closed circuit. More explicitly, for every edge \(r'v\in E(L)\setminus\{rr'\}\), with \(v\in\hat V\setminus\{r\}\), if \(rv\notin E(L)\), add the edge \(rv\); then delete the vertex \(r'\) and all edges incident to \(r'\). If the resulting graph contains a closed circuit, then this closed circuit contains the root vertex \(r\); otherwise it would already give a closed circuit in the tree \(L\), a contradiction. In this case, delete one edge of this closed circuit that is incident to \(r\) to give a tree. By Property~(P1), this deleted edge is not a positive-profit edge. Hence \(E(L)\cap\hat E^+=E(\Phi(L))\cap\hat E^+\). We denote the resulting rooted tree in \(\hat G\) by \(\Phi(L)\).

Let \(\mathcal T=\{\Phi(L):L\in\mathcal L^A\}\). For each \(T\in\mathcal T\), define
\[
\lambda_T=\sum_{L\in\mathcal L^A:\,\Phi(L)=T}\lambda^A_L .
\]
Then \(\lambda\) is a probability vector over \(\mathcal T\), where \(\lambda_T\) is the probability assigned to \(T\), and \(\sum_{T\in\mathcal T}\lambda_T=1\). Moreover, every \(T\in\mathcal T\) contains the root vertex \(r\), and \(\mathcal T\) is polynomial-size and computable in polynomial time.

Since \(w(e_0)=0\) and, for every edge \(rv\in \hat E\) incident to \(r\), the corresponding edge \(r'v\) in \(G^A\) satisfies \(w(r'v)=w(rv)\), the construction of \(\Phi(L)\) does not increase length for any \(L\in\mathcal L^A\). Therefore, sampling \(T\) according to \(\lambda\) and using Lemma~\ref{lem:treedecomp}, we obtain
\[
\mathbb E\left[\hat w(T)\right]
\le
\sum_{e\in\hat E}\hat w_e x_e .
\]
Moreover, since \(E(L)\cap\hat E^+=E(\Phi(L))\cap\hat E^+\) for every \(L\in\mathcal L^A\), Lemma~\ref{lem:treedecomp} gives
\[
\Pr[e\in E(T)]=x_e
\qquad \forall\,e\in\hat E^+ .
\]

It remains to prove the coupling between each positive-profit edge and its two endpoints. The construction of \(\Phi(L)\) does not change which vertices of \(\hat V\) are present. Hence \(\Pr[v\in V(T)]=y_v\) for every \(v\in\hat V\). Fix \(e=uv\in\hat E^+\). The coupling constraints~\eqref{eq:pcrpp-ip-couple} of \textsc{PCRPP-LP} give \(x_e=y_u=y_v\). Together with the edge probability above, this gives \(\Pr[e\in E(T)]=\Pr[u\in V(T)]=\Pr[v\in V(T)]\). Since the event \(e\in E(T)\) is contained in both events \(u\in V(T)\) and \(v\in V(T)\), no tree with positive probability can contain one endpoint of \(e\) without containing \(e\). Therefore, for every \(T\in\mathcal T\) with \(\lambda_T>0\),
\[
e\in E(T)
\quad\Longleftrightarrow\quad
u\in V(T)
\quad\Longleftrightarrow\quad
v\in V(T).
\]
\end{proof}

The next lemma isolates the effect of the outer-threshold splitting step in Algorithm~\ref{alg:best_of_many}. Although the preprocessed complete graph \(\hat G\) is not fully metric, the full sequence of complete splitting-off operations for a fixed threshold \(\delta\) preserves feasibility of \textsc{PCRPP-LP} and does not increase \(\sum_{e\in\hat E}\hat w_e x_e\).

\begin{lemma}\label{lem:low_y_splitting_pcrpp}
Fix an outer threshold \(\delta\) in Algorithm~\ref{alg:best_of_many}. After the complete splitting-off operations for this threshold have been performed, the resulting pair \((\tilde x,\tilde y)\) is feasible for \textup{\textsc{PCRPP-LP}}. Moreover,
\[
\sum_{e\in\hat E}\hat w_e\tilde x_e
\le
\sum_{e\in\hat E}\hat w_e x^*_e .
\]
For every \(v\in\hat V\setminus\{r\}\), the following two cases hold: if \(y_v^*<\delta\), then \(\tilde y_v=0\); if \(y_v^*\ge\delta\), then \(\tilde y_v=y_v^*\). Furthermore, for every positive-profit edge \(e\in\hat E^+\), we have
\[
\tilde x_e =
\begin{cases}
0, & x_e^*<\delta,\\
x_e^*, & x_e^*\ge \delta .
\end{cases}
\]
\end{lemma}

\begin{proof}
By the construction in Algorithm~\ref{alg:best_of_many}, if \(y_v^*<\delta\), then \(\tilde y_v=0\), while if \(y_v^*\ge\delta\), then \(\tilde y_v=y_v^*\). In particular, \(\tilde y_r=1\). We next verify that \((\tilde x,\tilde y)\) is feasible for \textsc{PCRPP-LP}.

We start with the degree constraints~\eqref{eq:pcrpp-ip-deg1}. Fix \(v\in\hat V\setminus\{r\}\) with \(y_v^*\ge\delta\). For the optimal solution \((x^*,y^*)\) to \textsc{PCRPP-LP}, the cut constraints~\eqref{eq:pcrpp-ip-cut} imply that the minimum \(v\)-to-\(r\) cut size is at least \(2y_v^*\), while the singleton cut \(\delta(v)\) has size \(x^*(\delta(v))=2y_v^*\) by the degree constraints~\eqref{eq:pcrpp-ip-deg1}. Hence this minimum cut size is exactly \(2y_v^*\). Since no complete splitting is performed at \(v\) or at \(r\), Lemma~\ref{lem:splitting} preserves this minimum cut size throughout the splitting sequence. Thus the final singleton cut \(\delta(v)\), being a \(v\)-to-\(r\) cut, has size at least \(2y_v^*\). On the other hand, the degree of \(v\) does not increase throughout the splitting sequence, so \(\tilde x(\delta(v))\le x^*(\delta(v))=2y_v^*\). Therefore \(\tilde x(\delta(v))=2y_v^*=2\tilde y_v\). Now consider \(v\in\hat V\setminus\{r\}\) with \(y_v^*<\delta\). If \(0<y_v^*<\delta\), then complete splitting is performed at \(v\), and the degree of \(v\) becomes zero after this complete splitting. If \(y_v^*=0\), then \(x^*(\delta(v))=0\) by the degree constraints~\eqref{eq:pcrpp-ip-deg1}, so the degree of \(v\) is already zero at the beginning of the splitting sequence. Complete splitting does not increase the degree of \(v\). Hence \(\tilde x(\delta(v))=0=2\tilde y_v\). Thus the degree constraints~\eqref{eq:pcrpp-ip-deg1} hold for every non-root vertex. The root-degree constraint~\eqref{eq:pcrpp-ip-deg2} is preserved because all complete splittings are performed at non-root vertices and do not increase the degree of \(r\).

We now verify the cut constraints~\eqref{eq:pcrpp-ip-cut}. Let \(S\subseteq\hat V\setminus\{r\}\) and \(v\in S\). If \(y_v^*<\delta\), then \(\tilde y_v=0\), and the constraint is immediate. If \(y_v^*\ge\delta\), then the argument above shows that the minimum \(v\)-to-\(r\) cut size with respect to \(\tilde x\) is at least \(2y_v^*=2\tilde y_v\). Since \(\delta(S)\) separates \(v\) from \(r\), we have \(\tilde x(\delta(S))\ge 2\tilde y_v\). Hence the cut constraints~\eqref{eq:pcrpp-ip-cut} hold.

It remains to verify the coupling constraints~\eqref{eq:pcrpp-ip-couple} for positive-profit edges. Let \(e=uv\in\hat E^+\). By the coupling constraints~\eqref{eq:pcrpp-ip-couple} for \((x^*,y^*)\), we have \(x_e^*=y_u^*=y_v^*\). If \(x_e^*<\delta\), then \(y_u^*<\delta\) and \(y_v^*<\delta\). The degree constraints~\eqref{eq:pcrpp-ip-deg1}, already proved for \((\tilde x,\tilde y)\), give \(\tilde x(\delta(u))=\tilde x(\delta(v))=0\). Hence \(\tilde x_e=0=\tilde y_u=\tilde y_v\). Now suppose that \(x_e^*\ge\delta\). Then \(y_u^*\ge\delta\) and \(y_v^*\ge\delta\), so no complete splitting is performed at \(u\) or at \(v\), and \(\tilde y_u=\tilde y_v=x_e^*\). Since neither endpoint of \(e\) is the vertex at which complete splitting is performed, \(\tilde x_e\ge x_e^*\). We prove the reverse inequality. If \(\tilde x_e>x_e^*\), then, since the degrees of \(u\) and \(v\) do not increase throughout the splitting sequence, we have \(\tilde x(\delta(u))\le x^*(\delta(u))=2x_e^*\) and \(\tilde x(\delta(v))\le x^*(\delta(v))=2x_e^*\). Hence the cut \(\delta(\{u,v\})\) has size at most \(\tilde x(\delta(u))+\tilde x(\delta(v))-2\tilde x_e \le 4x_e^*-2\tilde x_e<2x_e^*=2y_u^*\). By Property~(P1), \(r\notin\{u,v\}\), so \(\delta(\{u,v\})\) is a \(u\)-to-\(r\) cut. This contradicts Lemma~\ref{lem:splitting}, since the minimum \(u\)-to-\(r\) cut size is preserved throughout the splitting sequence and remains at least \(2y_u^*\). Therefore \(\tilde x_e\le x_e^*\), and hence \(\tilde x_e=x_e^*=\tilde y_u=\tilde y_v\). Thus the coupling constraints~\eqref{eq:pcrpp-ip-couple} hold for all positive-profit edges. Moreover, the same argument shows that \(\tilde x_e=0\) if \(x_e^*<\delta\), and \(\tilde x_e=x_e^*\) if \(x_e^*\ge\delta\), for every \(e\in\hat E^+\).

Constraint~\eqref{eq:pcrpp-ip-root}, and the relaxed bound constraints corresponding to~\eqref{eq:pcrpp-ip-x} and \eqref{eq:pcrpp-ip-y} are satisfied by the construction of \((\tilde x,\tilde y)\). Therefore \((\tilde x,\tilde y)\) is feasible for \textsc{PCRPP-LP}.

Finally, we prove the length inequality. For an edge \(e=uv\in\hat E\), let \(d_e\) be the shortest path distance between \(u\) and \(v\) in the graph \(\hat G\). Then \((\hat V,\hat E,d)\) is a metric complete graph. The complete splitting operations that produce \(\tilde x\) from \(x^*\) are performed on \(\hat G\) and do not depend on edge lengths; the edge lengths \(d\) are introduced only for the length analysis. Lemma~\ref{lem:splitting} gives \(\sum_{e\in\hat E}d_e\tilde x_e\le\sum_{e\in\hat E}d_e x_e^*\). By Property~(P3), \(\hat w_e=d_e\) for every zero-profit edge \(e\), and \(d_e\le\hat w_e\) for every positive-profit edge \(e\in\hat E^+\). Since the coupling argument above also gives \(\tilde x_e\le x_e^*\) for every positive-profit edge, we obtain
\[
\sum_{e\in\hat E}\hat w_e\tilde x_e
=
\sum_{e\in\hat E}d_e \tilde x_e
+
\sum_{e\in\hat E^+}(\hat w_e-d_e)\tilde x_e
\le
\sum_{e\in\hat E}d_e x_e^*
+
\sum_{e\in\hat E^+}(\hat w_e-d_e)x_e^*
=
\sum_{e\in\hat E}\hat w_e x_e^* .
\]
\end{proof}

We now analyze the candidate solutions generated by Algorithm~\ref{alg:best_of_many}. Lemma~\ref{lem:qjoin_cost_bound} and the proofs of Theorems~\ref{thm:golden_ratio} and \ref{thm:below_1_6} adapt the analyses of \cite{Blauth2026} from the PCTSP to the PCRPP. Fix an outer threshold \(\delta\), and let \((\tilde x,\tilde y)\) be the pair obtained after the complete splitting operations for this threshold \(\delta\). By Lemma~\ref{lem:low_y_splitting_pcrpp}, \((\tilde x,\tilde y)\) is feasible for \textsc{PCRPP-LP}. Apply Theorem~\ref{thm:pcrpp_treedecomp} to \((\tilde x,\tilde y)\), and let \(\mathcal T\) and \(\lambda\) be the resulting set of trees and probability vector. For \(T\in\mathcal T\) and an inner threshold \(\gamma\), we keep the notation \(H_{T,\gamma}\), \(J_{T,\gamma}\), and \(W_{T,\gamma}\) from Algorithm~\ref{alg:best_of_many}. The walk \(W_{T,\gamma}\) is feasible for the original PCRPP instance. The next two lemmas analyze, for fixed \(\delta\), \(T\), and \(\gamma\), the length of the restored multigraph \(H_{T,\gamma}\) and the length of the parity correction \(J_{T,\gamma}\), respectively. 

\begin{lemma}\label{lem:tree_cost_bound}
For every fixed \(\delta\), \(T\in\mathcal T\), and \(\gamma\),
\[
w(H_{T,\gamma})
\le
\hat w(\operatorname{core}_{\tilde x}(T,\gamma)) .
\]
Moreover, if \(T\) is sampled from \(\mathcal T\) according to \(\lambda\), then for every fixed \(\delta\) and \(\gamma\),
\[
\mathbb E[w(H_{T,\gamma})]
\le
\mathbb E\!\left[\hat w(T)\right]
\le
\sum_{e\in\hat E}\hat w_e\tilde x_e ,
\]
where the expectation is over the choice of \(T\).
\end{lemma}

\begin{proof}
Fix \(T\in\mathcal T\). Restoring \(\operatorname{core}_{\tilde x}(T,\gamma)\) to the original graph does not increase length. Each positive-profit edge is kept as the corresponding original edge with the same length, and each zero-profit edge is replaced by a shortest path whose length equals the length of the zero-profit edge in the preprocessed instance. Hence
\[
w(H_{T,\gamma})
\le
\hat w(\operatorname{core}_{\tilde x}(T,\gamma)) .
\]

Moreover, \(\operatorname{core}_{\tilde x}(T,\gamma)\) is a subtree of \(T\). Thus
\[
\hat w(\operatorname{core}_{\tilde x}(T,\gamma))
\le
\hat w(T) .
\]
Taking expectation over \(T\) sampled according to \(\lambda\), and applying the expected-length guarantee in Theorem~\ref{thm:pcrpp_treedecomp}, gives
\[
\mathbb E\!\left[
\hat w(\operatorname{core}_{\tilde x}(T,\gamma))
\right]
\le
\mathbb E\!\left[\hat w(T)\right]
\le
\sum_{e\in\hat E}\hat w_e\tilde x_e .
\]
Combining this with the fixed-\(T\) restoration bound proves the lemma.
\end{proof}

\begin{lemma}\label{lem:qjoin_cost_bound}
Fix \(\delta,\gamma\), and let \(T\in\mathcal T\) with \(\lambda_T>0\). Let \(\eta_1>\eta_2>\cdots>\eta_q\) be the distinct positive values among \(\{\tilde x_e:e\in\hat E^+\}\). For each \(i\in\{1,\ldots,q\}\), set \(T_i=\operatorname{core}_{\tilde x}(T,\eta_i)\), and let \(T_0\) be the trivial tree consisting only of the root. Then the minimum-length \(\operatorname{odd}(H_{T,\gamma})\)-join \(J_{T,\gamma}\) computed by the algorithm satisfies
\[
w(J_{T,\gamma})
\le
\frac{1}{3-\delta}\sum_{e\in\hat E}\hat w_e\tilde x_e
+
\sum_{i:\eta_i\ge\gamma}
\left(1-\frac{2\eta_i}{3-\delta}\right)
\hat w(E(T_i)\setminus E(T_{i-1})) .
\]
\end{lemma}

\begin{proof}
Write \(T_\gamma=\operatorname{core}_{\tilde x}(T,\gamma)\), and let \(\hat J\) be a minimum-length \(\operatorname{odd}(T_\gamma)\)-join in the preprocessed graph \(\hat G\). We restore \(\hat J\) to the original graph by replacing each zero-profit edge by the corresponding shortest path in \(G\), keeping each positive-profit edge as the corresponding original edge, and merging copied vertices. It is straightforward to check from this restoration that the resulting edge multiset has length at most \(\hat w(\hat J)\) and is an \(\operatorname{odd}(H_{T,\gamma})\)-join in \(G\). Since \(J_{T,\gamma}\) is a minimum-length \(\operatorname{odd}(H_{T,\gamma})\)-join in \(G\), we have \(w(J_{T,\gamma})\le \hat w(\hat J)\). Thus every upper bound on \(\hat w(\hat J)\) is also an upper bound on \(w(J_{T,\gamma})\).

Define \(z\in\mathbb R_{\ge0}^{\hat E}\) by
\[
z=
\frac{1}{3-\delta}\tilde x
+
\sum_{i:\eta_i\ge\gamma}
\left(1-\frac{2\eta_i}{3-\delta}\right)
\chi_{\hat E}^{E(T_i)\setminus E(T_{i-1})}.
\]
Since \(\eta_i\le1\) and \(\delta\le1\), we have \(1-2\eta_i/(3-\delta)\ge0\) for every \(i\) with \(\eta_i\ge\gamma\). We show that \(z\) is feasible for \textsc{Q-Join-LP} on \(\hat G\) with \(Q=\operatorname{odd}(T_\gamma)\). The nonnegativity constraints~\eqref{eq:qjoin-lp-nonneg} are immediate, so it remains to verify the cut constraints~\eqref{eq:qjoin-lp-cut}.

Let \(S\subseteq\hat V\) satisfy \(|S\cap\operatorname{odd}(T_\gamma)|\) odd. By the handshaking identity, \(|\delta_{T_\gamma}(S)|\) is odd. Replacing \(S\) by its complement if necessary, assume that \(r\notin S\).

First suppose that \(|\delta_{T_\gamma}(S)|\ge3\). Then \(S\) contains a non-root vertex \(v\) of \(T_\gamma\). We observe that every non-root vertex of \(T\) satisfies \(\tilde y_v\ge\delta\). Indeed, let \(v\in V(T)\setminus\{r\}\). Since \(\lambda_T>0\), the equality \(\Pr[v\in V(T)]=\tilde y_v\) in Theorem~\ref{thm:pcrpp_treedecomp}, applied to the pair \((\tilde x,\tilde y)\), implies \(\tilde y_v>0\). By Lemma~\ref{lem:low_y_splitting_pcrpp}, every non-root vertex has either \(\tilde y_v=0\) or \(\tilde y_v=y_v^*\ge\delta\). Hence \(\tilde y_v\ge\delta\). In particular, \(\tilde y_v\ge\delta\), and the cut constraints~\eqref{eq:pcrpp-ip-cut} give \(\tilde x(\delta(S))\ge2\delta\). Moreover,
\[
E(T_\gamma)
=
\bigcup_{i:\eta_i\ge\gamma}
\bigl(E(T_i)\setminus E(T_{i-1})\bigr).
\]
Hence
\[
\sum_{i:\eta_i\ge\gamma}
\chi_{\hat E}^{E(T_i)\setminus E(T_{i-1})}(\delta(S))
= |\delta_{T_\gamma}(S)| \ge 3 .
\]
Since \(\eta_i\le1\) for every \(i\), we have \(1-2\eta_i/(3-\delta)\ge 1-2/(3-\delta)\). Therefore
\[
z(\delta(S))
\ge
\frac{2\delta}{3-\delta}
+
\left(1-\frac{2}{3-\delta}\right)
\sum_{i:\eta_i\ge\gamma}
\chi_{\hat E}^{E(T_i)\setminus E(T_{i-1})}(\delta(S))
\ge
\frac{2\delta}{3-\delta}
+
3\left(1-\frac{2}{3-\delta}\right)
=1 .
\]
Thus constraint~\eqref{eq:qjoin-lp-cut} holds in this case.

Now suppose that \(|\delta_{T_\gamma}(S)|=1\). Let \(h\) be the unique edge of \(T_\gamma\) crossing \(S\), and let \(i\) be such that \(h\in E(T_i)\setminus E(T_{i-1})\). By the definition of the edge-profit core, \(T_i=\operatorname{core}_{\tilde x}(T,\eta_i)\) is the minimal subtree of \(T\) that contains \(r\) and contains every positive-profit edge \(e\in E(T)\cap\hat E^+\) with \(\tilde x_e\ge\eta_i\); equivalently, \(T_i\) is the union of the unique \(r\)-to-\(v\) paths in \(T\), where \(v\) is an endpoint of a positive-profit edge \(e\in E(T)\cap\hat E^+\) with \(\tilde x_e\ge\eta_i\). Since \(h\in E(T_i)\), there exist a positive-profit edge \(e^h\in E(T)\cap\hat E^+\) with \(\tilde x_{e^h}\ge\eta_i\) and an endpoint \(v^h\) of \(e^h\) such that the unique \(r\)-to-\(v^h\) path in \(T\) contains \(h\). This path is contained in \(T_i\subseteq T_\gamma\). Since \(r\notin S\) and \(h\) is the only edge of \(T_\gamma\) crossing \(S\), we must have \(v^h\in S\). By the coupling constraints~\eqref{eq:pcrpp-ip-couple}, \(\tilde y_{v^h}=\tilde x_{e^h} \ge\eta_i\). The cut constraints~\eqref{eq:pcrpp-ip-cut} give \(\tilde x(\delta(S))\ge2\tilde y_{v^h}\ge2\eta_i\). Moreover, \(\chi_{\hat E}^{E(T_i)\setminus E(T_{i-1})}(\delta(S))\ge1\), because \(h\in (E(T_i)\setminus E(T_{i-1}))\cap\delta(S)\). By the definition of \(z\), we obtain
\[
z(\delta(S))
\ge
\frac{2\eta_i}{3-\delta}
+
\left(1-\frac{2\eta_i}{3-\delta}\right)
=1 .
\]
Thus constraint~\eqref{eq:qjoin-lp-cut} also holds in this case.

Therefore \(z\) is feasible for \textsc{Q-Join-LP} on \(\hat G\) with \(Q=\operatorname{odd}(T_\gamma)\). By Lemma~\ref{lem:tjoin}, a minimum-length \(\operatorname{odd}(T_\gamma)\)-join in \(\hat G\) has length at most
\[
\sum_{e\in\hat E}\hat w_e z_e
=
\frac{1}{3-\delta}\sum_{e\in\hat E}\hat w_e\tilde x_e
+
\sum_{i:\eta_i\ge\gamma}
\left(1-\frac{2\eta_i}{3-\delta}\right)
\hat w(E(T_i)\setminus E(T_{i-1})) .
\]
As observed at the beginning of the proof, this upper bound on the minimum length of an \(\operatorname{odd}(T_\gamma)\)-join in \(\hat G\) is also an upper bound on \(w(J_{T,\gamma})\). This proves the lemma.
\end{proof}

We next use a sampling experiment over the candidate solutions generated by Algorithm~\ref{alg:best_of_many}. The sampling experiment is used only for the analysis: the random choices below select values of \(\delta\), \(T\), and \(\gamma\) whose corresponding candidate solutions are generated by the algorithm.

\begin{theorem}\label{thm:golden_ratio}
Algorithm~\ref{alg:best_of_many} is a \((1+\sqrt5)/2\)-approximation algorithm for the minimization version of the PCRPP.
\end{theorem}

\begin{proof}
Fix a threshold \(\delta\in[0,1)\) for the analysis. If there is a vertex \(v\in\hat V\setminus\{r\}\) with \(y_v^*\ge\delta\), then the vertices split off for this value of \(\delta\) are the same as those split off for the smallest value \(y_v^*\) that is at least \(\delta\), and this threshold is considered by Algorithm~\ref{alg:best_of_many}. If no vertex \(v\in\hat V\setminus\{r\}\) satisfies \(y_v^*\ge\delta\), then no nontrivial edge-profit core is obtained; the corresponding solution is the trivial rooted closed walk, which is included in Algorithm~\ref{alg:best_of_many} because \(W_A\) is initialized as the rooted closed walk that stays at \(r\). Thus every value of \(\delta\) used in the analysis either gives the same splitting result as a threshold considered by the algorithm, or gives only the trivial rooted closed walk already included by the algorithm.

Let \((\tilde x,\tilde y)\) be the pair obtained from \((x^*,y^*)\) after the complete splitting operations for this threshold \(\delta\). Apply Theorem~\ref{thm:pcrpp_treedecomp} to \((\tilde x,\tilde y)\), and let \(\mathcal T\) be the resulting set of trees with probability vector \(\lambda\). Fix \(\kappa\in[\delta,1]\). In the analysis, following \cite{Blauth2026}, choose \(T\in\mathcal T\) according to \(\lambda\), and then choose \(\gamma\) independently according to
\[
\Pr[\gamma\le t]=
\begin{cases}
0, & t<\delta,\\[2mm]
\dfrac{3-\delta-\kappa}{3-\delta-t}, & \delta\le t\le\kappa,\\[3mm]
1, & t>\kappa .
\end{cases}
\]

For a fixed tree \(T\), the core \(\operatorname{core}_{\tilde x}(T,\gamma)\) changes only when \(\gamma\) crosses a value \(\tilde x_e\) of an edge \(e\in E(T)\cap\hat E^+\). Therefore, by the same argument as for the outer threshold \(\delta\), every value of \(\gamma\) used in the analysis is represented either by an inner threshold considered by Algorithm~\ref{alg:best_of_many} or by the trivial rooted closed walk already included by the algorithm.

We first bound the expected length of \(W_{T,\gamma}\). For a fixed \(T\), use the notation of Lemma~\ref{lem:qjoin_cost_bound}: let \(\eta_1>\eta_2>\cdots>\eta_q\) be the distinct positive values among \(\{\tilde x_e:e\in\hat E^+\}\), set \(T_i=\operatorname{core}_{\tilde x}(T,\eta_i)\) for \(i\in\{1,\ldots,q\}\), and let \(T_0\) be the trivial tree consisting only of the root. The restoration bound used in Lemma~\ref{lem:tree_cost_bound}, together with Lemma~\ref{lem:qjoin_cost_bound}, gives, for fixed \(T\) and \(\gamma\),
\[
w(W_{T,\gamma}) \le w(H_{T,\gamma})+w(J_{T,\gamma})
\le
\frac{1}{3-\delta}\sum_{e\in\hat E}\hat w_e\tilde x_e
+
\sum_{i:\eta_i\ge\gamma}
\left(2-\frac{2\eta_i}{3-\delta}\right)
\hat w(E(T_i)\setminus E(T_{i-1})) .
\]
Taking expectation over \(\gamma\), with \(T\) fixed, gives
\[
\mathbb E\left[
\sum_{i:\eta_i\ge\gamma}
\left(2-\frac{2\eta_i}{3-\delta}\right)
\hat w(E(T_i)\setminus E(T_{i-1}))
\right]
=
\sum_{i=1}^q
\Pr[\gamma\le\eta_i]
\left(2-\frac{2\eta_i}{3-\delta}\right)
\hat w(E(T_i)\setminus E(T_{i-1})) .
\]
If \(\eta_i\in[\delta,\kappa]\), then
\[
\Pr[\gamma\le\eta_i]
\left(2-\frac{2\eta_i}{3-\delta}\right)
=
\frac{3-\delta-\kappa}{3-\delta-\eta_i}
\left(2-\frac{2\eta_i}{3-\delta}\right)
=
\frac{6-2\delta-2\kappa}{3-\delta}.
\]
If \(\eta_i\ge\kappa\), then
\[
\Pr[\gamma\le\eta_i]
\left(2-\frac{2\eta_i}{3-\delta}\right)
=
2-\frac{2\eta_i}{3-\delta}
\le
\frac{6-2\delta-2\kappa}{3-\delta},
\]
and if \(\eta_i<\delta\), then \(\Pr[\gamma\le\eta_i]=0\). Since the sets \(E(T_i)\setminus E(T_{i-1})\) are pairwise disjoint subsets of \(E(T)\), for every fixed \(T\) we obtain
\[
\mathbb E[w(H_{T,\gamma})+w(J_{T,\gamma})]
\le
\frac{1}{3-\delta}\sum_{e\in\hat E}\hat w_e\tilde x_e
+
\frac{6-2\delta-2\kappa}{3-\delta}
\hat w(T) ,
\]
where the expectation is over the choice of \(\gamma\). Taking expectation over \(T\), using Lemma~\ref{lem:tree_cost_bound}, and then using Lemma~\ref{lem:low_y_splitting_pcrpp}, gives
\[
\mathbb E[w(W_{T,\gamma})]
\le
\mathbb E[w(H_{T,\gamma})+w(J_{T,\gamma})]
\le
\frac{7-2\delta-2\kappa}{3-\delta}
\sum_{e\in\hat E}\hat w_e \tilde x_e
\le
\frac{7-2\delta-2\kappa}{3-\delta}
\sum_{e\in\hat E}\hat w_e x_e^* .
\]

We next bound the expected uncollected profit. If a positive-profit edge \(e\in\hat E^+\) belongs to \(E(\operatorname{core}_{\tilde x}(T,\gamma))\), then the restoration step places the corresponding original positive-profit edge in \(H_{T,\gamma}\), and therefore \(W_{T,\gamma}\) collects its profit. Hence the uncollected profit of \(W_{T,\gamma}\) is at most $\sum_{e\in\hat E^+\setminus E(\operatorname{core}_{\tilde x}(T,\gamma))} \hat p_e$.

Fix \(e\in\hat E^+\). If \(x_e^*<\delta\), then we use the trivial bound
\[
\Pr\!\left[
e\notin E(\operatorname{core}_{\tilde x}(T,\gamma))
\right]
\le
1
\le
\frac{1-x_e^*}{1-\delta}.
\]
Now suppose that \(x_e^*\ge\delta\). Then \(\tilde x_e=x_e^*\). By Theorem~\ref{thm:pcrpp_treedecomp}, the sampled tree contains \(e\) with probability \(x_e^*\). Moreover, if \(e\in E(T)\) and \(\gamma\le x_e^*\), then \(e\in E(\operatorname{core}_{\tilde x}(T,\gamma))\). Therefore, if \(x_e^*\in[\delta,\kappa]\), then
\[
\Pr\!\left[
e\notin E(\operatorname{core}_{\tilde x}(T,\gamma))
\right]
\le
1-
x_e^*\frac{3-\delta-\kappa}{3-\delta-x_e^*}
\le
\frac{3-\delta}{3-\delta-\kappa}(1-x_e^*).
\]
If \(x_e^*\ge\kappa\), then
\[
\Pr\!\left[
e\notin E(\operatorname{core}_{\tilde x}(T,\gamma))
\right]
\le
1-x_e^*
\le
\frac{3-\delta}{3-\delta-\kappa}(1-x_e^*).
\]
It follows that
\[
\mathbb E\left[
\sum_{e\in\hat E^+\setminus E(\operatorname{core}_{\tilde x}(T,\gamma))}
\hat p_e
\right]
\le
\max\left\{
\frac{1}{1-\delta},
\frac{3-\delta}{3-\delta-\kappa}
\right\}
\sum_{e\in\hat E^+}\hat p_e(1-x_e^*).
\]

The objective value \(A_{T,\gamma}\) of \(W_{T,\gamma}\) satisfies
\[
A_{T,\gamma}
\le
 w(W_{T,\gamma})
+
\sum_{e\in\hat E^+\setminus E(\operatorname{core}_{\tilde x}(T,\gamma))}
\hat p_e .
\]
Combining the expected length and the expected uncollected profit, we obtain
\[
\mathbb E[A_{T,\gamma}]
\le
\alpha(\delta,\kappa)
\left(
\sum_{e\in\hat E}\hat w_e x_e^*
+
\sum_{e\in\hat E^+}\hat p_e(1-x_e^*)
\right),
\]
where
\[
\alpha(\delta,\kappa)=
\max\left\{
\frac{7-2\delta-2\kappa}{3-\delta},
\frac{3-\delta}{3-\delta-\kappa},
\frac{1}{1-\delta}
\right\}.
\]
Set \(\kappa=1\) and \(\delta=(3-\sqrt5)/2\). Then
\[
\frac{7-2\delta-2\kappa}{3-\delta}
=
\frac{3-\delta}{3-\delta-\kappa}
=
\frac{1}{1-\delta}
=
\frac{1+\sqrt5}{2}.
\]
Thus \(\mathbb E[A_{T,\gamma}]\le (1+\sqrt5)\mathrm{OPT}_{\mathrm{LP}}/2\).

The random choices used in the analysis induce a probability distribution supported on candidate solutions generated by Algorithm~\ref{alg:best_of_many}. Since the algorithm returns a candidate solution with minimum objective value among the generated candidate solutions, \(\mathrm{ALG}\le \mathbb E[A_{T,\gamma}]\). By Theorem~\ref{thm:lp_bound}, \(\mathrm{OPT}_{\mathrm{LP}}\le\mathrm{OPT}\), and therefore \(\mathrm{ALG}\le (1+\sqrt5)\mathrm{OPT}/2\).
\end{proof}

In the preceding analysis, the outer threshold was fixed at \(\delta=(3-\sqrt5)/2\). We now refine the analysis by also randomizing the outer threshold \(\delta\).

\begin{theorem}\label{thm:below_1_6}
Algorithm~\ref{alg:best_of_many} is a better-than-\(1.6\)-approximation algorithm for the minimization version of the PCRPP.
\end{theorem}

\begin{proof}
Let \(0\le\kappa_0<\kappa\le1\) and \(\beta>0\) be parameters to be chosen later. In the analysis, choose \(\delta\in[\kappa_0,\kappa]\) with density
\[
f(\delta)=\nu(3-\delta)(\kappa-\delta)^\beta,
\qquad
\nu^{-1}=
\frac{(3-\kappa)(\kappa-\kappa_0)^{\beta+1}}{\beta+1}
+
\frac{(\kappa-\kappa_0)^{\beta+2}}{\beta+2}.
\]
For a fixed value of \(\delta\), the analysis in the proof of Theorem~\ref{thm:golden_ratio} gives
\[
\mathbb E[w(W_{T,\gamma})]\le \mathbb E\!\left[w(H_{T,\gamma})+w(J_{T,\gamma})\right]
\le
\frac{7-2\delta-2\kappa}{3-\delta}
\sum_{e\in\hat E}\hat w_e x_e^* .
\]
Taking expectation over \(\delta\) gives
\[
\mathbb E[w(W_{T,\gamma})]
\le
g(\kappa,\kappa_0)
\sum_{e\in\hat E}\hat w_e x_e^* ,
\]
where
\[
g(\kappa,\kappa_0)=
\nu\left(
\frac{(7-4\kappa)(\kappa-\kappa_0)^{\beta+1}}{\beta+1}
+
\frac{2(\kappa-\kappa_0)^{\beta+2}}{\beta+2}
\right).
\]

We next bound the expected uncollected profit as in Theorem~\ref{thm:golden_ratio}. Fix \(e\in\hat E^+\). If \(x_e^*<\kappa_0\), then the trivial bound gives
\[
\Pr\!\left[
e\notin E(\operatorname{core}_{\tilde x}(T,\gamma))
\right]
\le
1
\le
\frac{1-x_e^*}{1-\kappa_0}.
\]
If \(x_e^*\ge\kappa\), then for every sampled value of \(\delta\) we have \(\delta\le\kappa\le x_e^*\), and hence Lemma~\ref{lem:low_y_splitting_pcrpp} gives \(\tilde x_e=x_e^*\). Also \(\gamma\le\kappa\le x_e^*\). Therefore, whenever \(e\in E(T)\), we have \(e\in E(\operatorname{core}_{\tilde x}(T,\gamma))\). Since Theorem~\ref{thm:pcrpp_treedecomp} gives \(\Pr[e\in E(T)]=x_e^*\), it follows that
\[
\Pr\!\left[
e\notin E(\operatorname{core}_{\tilde x}(T,\gamma))
\right]
\le
1-x_e^* .
\]

It remains to consider \(x_e^*\in[\kappa_0,\kappa]\). For a sampled value \(\delta\le x_e^*\), Lemma~\ref{lem:low_y_splitting_pcrpp} gives \(\tilde x_e=x_e^*\). Hence, whenever \(e\in E(T)\) and \(\gamma\le x_e^*\), the edge \(e\) belongs to \(\operatorname{core}_{\tilde x}(T,\gamma)\). For this fixed value of \(\delta\), Theorem~\ref{thm:pcrpp_treedecomp} gives \(\Pr[e\in E(T)]=x_e^*\), and the distribution of \(\gamma\) gives $\Pr[\gamma\le x_e^*] =(3-\delta-\kappa)/(3-\delta-x_e^*)$. Integrating over sampled values of \(\kappa_0\le \delta\le x_e^*\), we obtain
\[
\Pr\!\left[
e\in E(\operatorname{core}_{\tilde x}(T,\gamma))
\right]
\ge
x_e^*\nu
\int_{\kappa_0}^{x_e^*}
(\kappa-\delta)^\beta
\frac{(3-\delta)(3-\delta-\kappa)}{3-\delta-x_e^*}
\,d\delta .
\]

For \(\xi\in[\kappa_0,\kappa]\), define
\[
\phi_\xi(\delta)=
\frac{(3-\delta-\kappa)(3-\delta)}{3-\delta-\xi}.
\]
For every \(\xi\le\kappa\), the function \(\phi_\xi\) is concave on \([\kappa_0,\kappa]\), since
\[
\phi_\xi(\delta)=3-\delta-\kappa+\xi
-\frac{\xi(\kappa-\xi)}{3-\delta-\xi},
\qquad
\phi_\xi''(\delta)=
-\frac{2\xi(\kappa-\xi)}{(3-\delta-\xi)^3}\le0.
\]
Hence, for every \(\delta\in[\kappa_0,\kappa]\),
\[
\phi_\xi(\delta)
\ge
\frac{\kappa-\delta}{\kappa-\kappa_0}\phi_\xi(\kappa_0)
+
\frac{\delta-\kappa_0}{\kappa-\kappa_0}\phi_\xi(\kappa).
\]
Applying this inequality with \(\xi=x_e^*\) and evaluating the resulting integrals gives
\[
\Pr\!\left[
e\notin E(\operatorname{core}_{\tilde x}(T,\gamma))
\right]
\le
h_{x_e^*}(\kappa,\kappa_0),
\]
where, for \(\xi\in[\kappa_0,\kappa]\),
\[
\begin{aligned}
h_\xi(\kappa,\kappa_0)=
1-\frac{\xi\nu}{\kappa-\kappa_0}
\Bigg[
&(\phi_\xi(\kappa_0)-\phi_\xi(\kappa))
\frac{(\kappa-\kappa_0)^{\beta+2}-(\kappa-\xi)^{\beta+2}}{\beta+2}\\
&+\phi_\xi(\kappa)(\kappa-\kappa_0)
\frac{(\kappa-\kappa_0)^{\beta+1}-(\kappa-\xi)^{\beta+1}}{\beta+1}
\Bigg].
\end{aligned}
\]

Combining the three cases, and using \(1/(1-\kappa_0)\ge1\), for every \(e\in\hat E^+\) we have
\[
\Pr\!\left[
e\notin E(\operatorname{core}_{\tilde x}(T,\gamma))
\right]
\le
\max\left\{
\frac{1}{1-\kappa_0},
\max_{\xi\in[\kappa_0,\kappa]}
\frac{h_\xi(\kappa,\kappa_0)}{1-\xi}
\right\}
(1-x_e^*).
\]
Therefore
\[
\mathbb E\left[
\sum_{e\in\hat E^+\setminus E(\operatorname{core}_{\tilde x}(T,\gamma))}
\hat p_e
\right]
\le
\max\left\{
\frac{1}{1-\kappa_0},
\max_{\xi\in[\kappa_0,\kappa]}
\frac{h_\xi(\kappa,\kappa_0)}{1-\xi}
\right\}
\sum_{e\in\hat E^+}\hat p_e(1-x_e^*).
\]

By the definition of \(A_{T,\gamma}\), $A_{T,\gamma}
\le
w(W_{T,\gamma})
+
\sum_{e\in\hat E^+\setminus E(\operatorname{core}_{\tilde x}(T,\gamma))}
\hat p_e$. Combining the expected length and the expected uncollected profit, we obtain
\[
\mathbb E[A_{T,\gamma}]
\le
\alpha(\kappa_0,\kappa,\beta)
\left(
\sum_{e\in\hat E}\hat w_e x_e^*
+
\sum_{e\in\hat E^+}\hat p_e(1-x_e^*)
\right),
\]
where
\[
\alpha(\kappa_0,\kappa,\beta)
=
\max\left\{
g(\kappa,\kappa_0),
\frac{1}{1-\kappa_0},
\max_{\xi\in[\kappa_0,\kappa]}
\frac{h_\xi(\kappa,\kappa_0)}{1-\xi}
\right\}.
\]

A numerical search for small values of \(\alpha(\kappa_0,\kappa,\beta)\) gives the following parameter choice:
\[
\kappa_0=0.36621005,\qquad
\kappa=0.99678328,\qquad
\beta=1.98094420.
\]
For these values, direct computation gives the upper bounds
\[
g(\kappa,\kappa_0)<1.59862255,
\qquad
\frac{1}{1-\kappa_0}<1.57780982 .
\]
By Lemma~\ref{lem:boundanalysis} in Appendix~\ref{app:boundanalysis},
\[
\max_{\xi\in[\kappa_0,\kappa]}
\frac{h_\xi(\kappa,\kappa_0)}{1-\xi}
<1.59872206.
\]
This proves that $\alpha(0.36621005,0.99678328,1.98094420)<1.6$. Consequently,
\[
\mathbb E[A_{T,\gamma}]
\le
\alpha(0.36621005,0.99678328,1.98094420)\,
\mathrm{OPT}_{\mathrm{LP}} .
\]

As in the proof of Theorem~\ref{thm:golden_ratio}, the random choices used in the analysis induce a probability distribution supported on candidate solutions generated by Algorithm~\ref{alg:best_of_many}. Since the algorithm returns a candidate solution with minimum objective value among the generated candidate solutions, \(\mathrm{ALG}\le\mathbb E[A_{T,\gamma}]\). By Theorem~\ref{thm:lp_bound}, \(\mathrm{OPT}_{\mathrm{LP}}\le\mathrm{OPT}\), and therefore
\[
\mathrm{ALG}
\le
\alpha(0.36621005,0.99678328,1.98094420)\,\mathrm{OPT}.
\]
Since \(\alpha(0.36621005,0.99678328,1.98094420)<1.6\), this gives a polynomial time approximation algorithm for the PCRPP with an approximation ratio strictly smaller than \(1.6\).
\end{proof}

\section{Computational experiments}
\label{sec:experiments}

This section reports a computational evaluation of Algorithm~\ref{alg:best_of_many}. Since our focus is on approximation algorithms with worst-case guarantees, we compare Algorithm~\ref{alg:best_of_many} with the reported exact optima on the benchmark set and with another approximation algorithm that also has a worst-case guarantee. For the latter comparison, we use the PCTSP-reduction algorithm analyzed in Appendix~\ref{app:pctsp-reduction}. The experiments evaluate the quality of the solutions returned by Algorithm~\ref{alg:best_of_many} and its performance relative to the PCTSP-reduction algorithm.

\subsection{Benchmark instances and objective conversion}
\label{sec:benchmark}
We use the 118 benchmark instances of \citet{Araoz2009}. These instances are grouped into five families: ALBAIDA, D, G, P, and R. The exact values reported in \citet{Araoz2009} are for the maximization objective \eqref{eq:obj_araoz}, whereas all approximation guarantees in this paper are stated for the minimization objective \eqref{eq:obj_min}. Hence, for each instance \(I\), we convert the reported optimum by
\[
 \mathrm{OPT}(I)
    =
 \sum_{e\in E(I)} p_e
    -
 \mathrm{OPT}_{\max}(I),
\]
where \(\mathrm{OPT}_{\max}(I)\) is the reported optimal value for \eqref{eq:obj_araoz}.

We denote by \(\mathrm{ALG}(I)\) the value returned by Algorithm~\ref{alg:best_of_many}. We denote by \(\mathrm{OPT}_{\mathrm{LP}}(I)\) the optimum value of the \textsc{PCRPP-LP} used in this paper. It is important to emphasize that \(\mathrm{OPT}_{\mathrm{LP}}(I)\) is not the value of a linear programming relaxation written directly on the original input graph. It is the optimum value of the canonical \textsc{PCRPP-LP} on the preprocessed complete graph constructed in Section~\ref{sec:preprocessing}. This is the lower bound on the optimum value used in the analysis of the approximation algorithm.

For each instance, we report the final optimality gap
\[
 \mathrm{ALG~gap}(I)
    =
    100\cdot
 \frac{\mathrm{ALG}(I)-\mathrm{OPT}(I)}
         {\mathrm{OPT}(I)}
\]
and the LP gap
\[
 \mathrm{LP~gap}(I)
    =
    100\cdot
 \frac{\mathrm{OPT}(I)-\mathrm{OPT}_{\mathrm{LP}}(I)}
         {\mathrm{OPT}(I)}.
\]
The first quantity measures the quality of the solution returned by Algorithm~\ref{alg:best_of_many}, against the reported exact optimum. The second measures the strength of the \textsc{PCRPP-LP} lower bound on the preprocessed complete graph.

The detailed per-instance results are reported in Appendix~\ref{app:computational-results}. That table includes, for every instance, the numbers of vertices and edges, the converted optimum \(\mathrm{OPT}\), the value \(\mathrm{ALG}\) returned by Algorithm~\ref{alg:best_of_many}, the value \(\mathrm{RED}\) returned by the PCTSP-reduction algorithm in Appendix~\ref{app:pctsp-reduction}, the LP value \(\mathrm{OPT}_{\mathrm{LP}}\), the three relative gaps, and the total running times of the two algorithms. In the appendix table, the better value between \(\mathrm{ALG}\) and \(\mathrm{RED}\) is shown in bold; if the two values are equal, both are shown in bold.

\subsection{Implementation details}
\label{sec:implementation}
All experiments were run in Python on a computer with an AMD Ryzen 7 5800H CPU and 16 GB RAM. In the implementation of Algorithm~\ref{alg:best_of_many}, we solve the linear programming relaxation by a cutting-plane method, using HiGHS through its Python interface. The cut constraints~\eqref{eq:pcrpp-ip-cut} in \textsc{PCRPP-LP} are exponential in number, and the preprocessed complete graph may contain many edge variables. We therefore start with a linear program containing the non-cut constraints on an initial set of edge variables, and add both violated cut constraints and missing edge variables dynamically. After solving the current linear program, we separate the cut constraints by min-cut computations and check omitted edge variables by reduced-cost pricing. Whenever violated constraints of the form~\eqref{eq:pcrpp-ip-cut} are found, or omitted edge variables with negative reduced cost are found, we add them to the current linear program and solve it again. We repeat this process until no violated cut constraint remains and no omitted edge variable has negative reduced cost, up to a numerical tolerance. Since separation and pricing are both performed after each reoptimization, the final current linear program satisfies all cut constraints and no omitted edge variable can improve the objective. The resulting solution is therefore an optimal solution to \textsc{PCRPP-LP} up to a numerical tolerance.

For the complete splitting routines, we use a constructive procedure based on min-cut computations. As in the proof of Lemma~\ref{lem:low_y_splitting_pcrpp}, the feasible splitting operations in our setting are non-degenerate. Indeed, when splitting off a vertex \(v\), a degenerate operation would choose a single incident edge \(vu\) and decrease its value. Since \(u\) is not being split off, \(y_u\) is kept fixed, while the size of the singleton cut \(\delta(u)\), and hence the required \(r\)-to-\(u\) cut size, decreases. Thus such an operation cannot preserve the required \(r\)-to-\(u\) cut size. The implementation therefore scans choices of two incident edges \(vu\) and \(vw\) that have positive value in the current edge vector. For a fixed choice of \(vu\) and \(vw\), the largest admissible splitting amount \(\varepsilon\) is computed by min-cut computations so that the required cut sizes between \(r\) and all vertices other than \(r\) and \(v\) are preserved. This is sufficient for maintaining feasibility of \(\textsc{PCRPP-LP}\), although the general splitting-off statement can be formulated with preservation of all pairwise minimum cut sizes among vertices other than \(v\). Once a choice with positive admissible amount \(\varepsilon\) is found, the split is performed immediately. The operation decreases the current values on \(vu\) and \(vw\) by \(\varepsilon\), and increases the current value on \(uw\) by \(\varepsilon\). This process is repeated until the current degree \(x(\delta(v))\) becomes zero up to a numerical tolerance.

Another implementation detail concerns the tree decompositions for different outer thresholds \(\delta\). In Algorithm~\ref{alg:best_of_many}, after an optimal solution \((x^*,y^*)\) of \textsc{PCRPP-LP} is obtained, each outer threshold \(\delta\) determines the vertices with \(y_v^*<\delta\) that are split off, resulting in the pair \((\tilde x,\tilde y)\). For this fixed threshold, Theorem~\ref{thm:pcrpp_treedecomp} is applied to \((\tilde x,\tilde y)\). In the proof of that theorem, the auxiliary complete graph \(G^A\) is constructed, and the vectors \((\bar x,\bar y)\) on \(G^A\) are defined from \((\tilde x,\tilde y)\) by the formulas in that proof. Then complete splittings are performed in \(G^A\) at all vertices except \(r\) and \(r'\), where \(e_0=rr'\) is the auxiliary edge, in nondecreasing order of the corresponding values \(y_v^*\). The resulting splitting-off operations are undone in reverse order to construct trees in \(G^A\), following the proof of Lemma~5 in \citet{Blauth2026}. These trees are finally converted into trees in \(\hat G\). In the implementation, we do not apply the edge-profit tree decomposition independently for every threshold \(\delta\). Instead, before fixing \(\delta\), we construct \(G^A\) and define \((\bar x,\bar y)\) from \((x^*,y^*)\) by the same formulas used in the proof of Theorem~\ref{thm:pcrpp_treedecomp}. We then perform complete splittings in \(G^A\) at all vertices except \(r\) and \(r'\), in the same nondecreasing order of the values \(y_v^*\), and record the resulting splitting-off operations once. For a fixed threshold \(\delta\), we undo this recorded sequence in reverse order and stop at the position determined by \(\delta\). We then construct the corresponding trees in \(G^A\), and convert these trees into trees in \(\hat G\) as in the proof of Theorem~\ref{thm:pcrpp_treedecomp}. Since different choices of \(\delta\) only determine where this same recorded reverse-order process is stopped, it is easy to verify from the construction that the resulting trees in \(\hat G\) are the same as those obtained by first splitting off the vertices with \(y_v^*<\delta\) to obtain \((\tilde x,\tilde y)\) and then applying Theorem~\ref{thm:pcrpp_treedecomp} to this pair \((\tilde x,\tilde y)\). Thus, this implementation avoids repeating the same splitting-off work for every threshold \(\delta\). It only improves the running time and does not change the candidate solutions considered by Algorithm~\ref{alg:best_of_many}.

The min-cut computations used for cut separation and for complete splitting are implemented by a standard maximum-flow routine. For parity correction, we solve the minimum-length \(Q\)-join problem via a minimum-weight perfect matching computation, using NetworkX's matching routine.

\subsection{Solution quality}
\label{sec:solution_quality}
Table~\ref{tab:quality-summary} summarizes the solution quality and the gap of \textsc{PCRPP-LP}. Algorithm~\ref{alg:best_of_many} finds an optimal solution on 34 of the 118 instances. The optimality gap of Algorithm~\ref{alg:best_of_many} is at most \(5\%\) on 81 instances and at most \(10\%\) on 113 instances. The average optimality gap is \(3.39\%\), and the maximum optimality gap is \(12.12\%\), attained on instance G22. The average gap of \textsc{PCRPP-LP} is \(2.09\%\), and the maximum gap of \textsc{PCRPP-LP} is \(7.17\%\), attained on instance G33.

\begin{table}
\centering
\caption{Solution quality and \(\mathrm{LP~gap}\). The columns ``Opt.'', ``\(\le 5\%\)'', and ``\(\le 10\%\)'' count instances according to \(\mathrm{ALG~gap}\). All gap columns are percentages.}
\label{tab:quality-summary}
\footnotesize
\setlength{\tabcolsep}{3.5pt}
\begin{tabular}{lrrrrrrrr}
\toprule
Family
& \# inst.
& Opt.
& \(\le 5\%\)
& \(\le 10\%\)
& Avg. \(\mathrm{ALG~gap}\)
& Max \(\mathrm{ALG~gap}\)
& Avg. \(\mathrm{LP~gap}\)
& Max \(\mathrm{LP~gap}\) \\
\midrule
ALBAIDA & 2   & 0  & 1  & 2   & 5.65 & 6.94  & 4.15 & 4.70 \\
D       & 36  & 6  & 21 & 36  & 4.23 & 9.67  & 2.45 & 6.07 \\
G       & 36  & 15 & 20 & 31  & 4.22 & 12.12 & 2.45 & 7.17 \\
P       & 24  & 7  & 19 & 24  & 2.67 & 8.42  & 2.40 & 5.77 \\
R       & 20  & 6  & 20 & 20  & 1.01 & 3.44  & 0.19 & 0.66 \\
\midrule
All     & 118 & 34 & 81 & 113 & 3.39 & 12.12 & 2.09 & 7.17 \\
\bottomrule
\end{tabular}
\end{table}

\subsection{Comparison with an algorithm with a worst-case guarantee}
\label{sec:comparison}
Since Algorithm~\ref{alg:best_of_many} is an approximation algorithm with a worst-case guarantee, we compare it with another approximation algorithm that also has a worst-case guarantee. For this purpose, we use the PCTSP-reduction algorithm obtained from the reduction described in Section~\ref{sec:contributions}. In the experiments, this reduction algorithm uses the better-than-\(1.6\)-approximation algorithm of \citet{Blauth2026} for the reduced PCTSP instance. By the analysis in Appendix~\ref{app:pctsp-reduction}, the resulting PCTSP-reduction algorithm has a better-than-\(3.2\)-approximation ratio for the PCRPP. In our implementation of this PCTSP-reduction algorithm, we run the algorithm of \citet{Blauth2026} using the same computational conventions as in Section~\ref{sec:implementation}: the linear programming relaxation, splitting operations, minimum-cut computations, and minimum-length \(Q\)-join computations are handled in the same way, and the tree decomposition is implemented in the same style.

For each instance \(I\), we denote by \(\mathrm{RED}(I)\) the objective value of the walk returned by this PCTSP-reduction algorithm, and define
\[
 \mathrm{RED~gap}(I)
    =
    100\cdot
 \frac{\mathrm{RED}(I)-\mathrm{OPT}(I)}
         {\mathrm{OPT}(I)}.
\]
Thus, \(\mathrm{RED}(I)\) is the value used to evaluate the solution quality of this reduction-based approximation algorithm on the benchmark instances.

Table~\ref{tab:reduction-summary} compares Algorithm~\ref{alg:best_of_many} with the PCTSP-reduction algorithm. The PCTSP-reduction algorithm has an average optimality gap of \(8.90\%\) and a maximum optimality gap of \(26.83\%\), attained on instance G16, while the corresponding values for Algorithm~\ref{alg:best_of_many} are \(3.39\%\) and \(12.12\%\). Algorithm~\ref{alg:best_of_many} returns a strictly smaller objective value on 101 instances. The PCTSP-reduction algorithm returns a strictly smaller objective value on 5 instances, and the two methods tie on 12 instances. These results should be interpreted as a comparison between two approximation algorithms with worst-case guarantees. The PCTSP-reduction algorithm converts edge profits into vertex penalties, whereas Algorithm~\ref{alg:best_of_many} treats edge profits directly. On these benchmark instances, this direct treatment of edge profits gives smaller objective values on most instances.

\begin{table}
\centering
\caption{Comparison with the PCTSP-reduction algorithm. The gap columns are computed according to \(\mathrm{RED~gap}\). The last three columns compare the objective values \(\mathrm{ALG}\) and \(\mathrm{RED}\).}
\label{tab:reduction-summary}
\small
\begin{tabular}{lrrrrrr}
\toprule
Family
& \# inst.
& Avg. RED gap
& Max RED gap
& ALG better
& RED better
& Tie \\
\midrule
ALBAIDA & 2   & 11.30 & 11.79 & 2   & 0 & 0 \\
D       & 36  & 8.48  & 20.81 & 29  & 4 & 3 \\
G       & 36  & 11.04 & 26.83 & 27  & 1 & 8 \\
P       & 24  & 9.22  & 19.05 & 23  & 0 & 1 \\
R       & 20  & 5.19  & 9.43  & 20  & 0 & 0 \\
\midrule
All     & 118 & 8.90  & 26.83 & 101 & 5 & 12 \\
\bottomrule
\end{tabular}
\end{table}

\subsection{Running time}
\label{sec:running_time}
Table~\ref{tab:runtime-summary} reports average and maximum running times. The average total running time of Algorithm~\ref{alg:best_of_many} is \(573.43\) seconds, and its maximum total running time is \(8295.91\) seconds, attained on instance D35. The average total running time of the PCTSP-reduction algorithm is \(287.12\) seconds, and its maximum total running time is \(8159.58\) seconds, also attained on instance D35. For Algorithm~\ref{alg:best_of_many}, the average time spent on solving the linear programming relaxation is \(165.64\) seconds, the average time spent on complete splitting-off is \(403.71\) seconds, and the average remaining time is \(4.08\) seconds. Here ``Other'' means the total running time of Algorithm~\ref{alg:best_of_many} minus its linear programming time and complete splitting-off time. All running times are measured in seconds.

\begin{table}
\centering
\caption{Running-time summary on an AMD Ryzen 7 5800H machine with 16 GB RAM. Times are in seconds.}
\label{tab:runtime-summary}
\footnotesize
\setlength{\tabcolsep}{3.5pt}
\begin{tabular}{lrrrrrrr}
\toprule
Family
& Avg. ALG
& Max ALG
& Avg. lin. prog.
& Avg. split
& Avg. other
& Avg. RED
& Max RED \\
\midrule
ALBAIDA & 2337.32 & 2717.35 & 672.02 & 1655.18 & 10.12 & 771.42 & 1176.25 \\
D       & 1304.74 & 8295.91 & 369.71 & 927.99  & 7.04  & 685.13 & 8159.58 \\
G       & 127.15  & 1198.95 & 17.39  & 108.60  & 1.16  & 35.69  & 305.46 \\
P       & 162.75  & 1291.73 & 27.08  & 134.13  & 1.54  & 54.39  & 578.36 \\
R       & 376.78  & 2193.78 & 180.78 & 189.51  & 6.48  & 254.10 & 1118.60 \\
\midrule
All     & 573.43  & 8295.91 & 165.64 & 403.71  & 4.08  & 287.12 & 8159.58 \\
\bottomrule
\end{tabular}
\end{table}

The running time of Algorithm~\ref{alg:best_of_many} is mainly spent on solving the linear programming relaxation and on complete splitting-off. As explained in Section~\ref{sec:implementation}, the construction of the edge-profit tree decompositions relies on complete splitting-off operations. In our implementation, these operations require a polynomial number of maximum-flow or minimum-cut computations, and this is the dominant component of the running time apart from solving the linear programming relaxation in Table~\ref{tab:runtime-summary}. The remaining time includes the preprocessing that constructs the preprocessed complete graph \(\hat G\), constructing the trees from the recorded splitting-off operations, the pruning step, the relevant \(Q\)-join computations, and the evaluation of candidate solutions. Algorithm~\ref{alg:best_of_many} is slower than the PCTSP-reduction algorithm on average mainly because the two algorithms work on complete graphs of different sizes. For the same original instance, Algorithm~\ref{alg:best_of_many} constructs the preprocessed complete graph \(\hat G\) on the vertex set obtained after the vertex-copying step in Section~\ref{sec:preprocessing}. The PCTSP-reduction algorithm instead constructs the metric complete graph \(G^R\) only on the root and the representative vertices corresponding to positive-profit edges. Thus \(\hat G\) typically has many more vertices and many more edges than \(G^R\). This leads to a larger linear programming relaxation and substantially more splitting-off work, which accounts for the larger running time of Algorithm~\ref{alg:best_of_many}.

Overall, the experiments indicate that Algorithm~\ref{alg:best_of_many} can be implemented on the benchmark instances and that, on most instances, the objective values of the returned solutions are close to the reported exact optima. Compared with the PCTSP-reduction algorithm, which is also an approximation algorithm with a worst-case guarantee, Algorithm~\ref{alg:best_of_many} returns better objective values on most instances. The running-time results show that the main computational work is solving the linear programming relaxation and performing complete splitting-off, while the remaining steps account for only a small fraction of the total running time.

\section{Conclusion}
\label{sec:conclusion}

We study the PCRPP, a variant of the rural postman problem. The original formulation maximizes the collected profit minus the walk length. In this paper, we analyze the corresponding minimization version, whose objective value is the walk length plus the total profit of the edges not traversed by the walk. For the same instance, this minimization version has the same optimal solutions as the original maximization formulation.

Our main result is a polynomial time approximation algorithm with an approximation ratio strictly smaller than \(1.6\) for the minimization version of the PCRPP. The algorithm treats edge profits directly, rather than converting them into vertex penalties as in the PCTSP-reduction approach. It is based on a linear programming relaxation, splitting-off techniques, an edge-profit tree decomposition, pruning, and parity correction. The computational results show that the objective values of the returned walks are close to the reported exact optimum values on most instances. The algorithm also returns better solutions on most instances than the PCTSP-reduction algorithm that uses the PCTSP algorithm of \citet{Blauth2026} after converting edge profits into vertex penalties.

\appendix

\section{The PCTSP-reduction algorithm}
\label{app:pctsp-reduction}

We use the reduced PCTSP instance constructed in Section~\ref{sec:contributions}. Let \(G^R\) be the complete metric graph in that construction, and let \(R=\{r\}\cup\{s_e:p_e>0\}\) be its vertex set, where \(s_e\) is the representative vertex corresponding to the positive-profit edge \(e\). We describe the PCTSP-reduction algorithm and prove its \(2\rho\)-approximation guarantee, where \(\rho\) denotes the approximation ratio of the PCTSP algorithm applied to this reduced PCTSP instance. The PCTSP-reduction algorithm first applies the \(\rho\)-approximation algorithm for the PCTSP to \(G^R\). Suppose that the returned solution to the PCTSP instance visits the representative vertices \(s_{e_1},s_{e_2},\ldots,s_{e_k}\) in cyclic order, where \(e_i=u_iv_i\). We then construct a feasible solution to the PCRPP instance that traverses these selected positive-profit edges in this order. If \(k=0\), we output the rooted closed walk containing only \(r\). Otherwise, set \(t_0=r\). For each \(i=1,\ldots,k\), choose one of the two endpoints of \(e_i\) whose shortest-path distance from \(t_{i-1}\) in the original graph is minimum, move from \(t_{i-1}\) to this chosen endpoint along a shortest path, traverse \(e_i\), and let \(t_i\) be the endpoint reached after this traversal.
 After traversing \(e_k\), we return from \(t_k\) to \(r\) along a shortest path.

It remains to bound the objective value of the resulting walk. It is easy to verify from the construction of the reduced PCTSP instance that its optimum value is at most the optimum value of the original PCRPP instance, because any feasible solution to the PCRPP instance naturally gives a feasible solution to the reduced PCTSP instance with no larger objective value. We therefore only compare the returned solution to the PCTSP instance on \(G^R\) with the feasible solution to the PCRPP instance constructed from it. The case \(k=0\) is immediate, so assume \(k\ge1\). Let the returned solution to the PCTSP instance have length \(L\). Let \(D_0\) be the length of \(rs_{e_1}\) in \(G^R\), let \(D_i\) be the length of \(s_{e_i}s_{e_{i+1}}\) in \(G^R\) for \(i=1,\ldots,k-1\), and let \(D_k\) be the length of \(s_{e_k}r\) in \(G^R\). Then \(L=D_0+\cdots+D_k\). Since the edge \(rs_{e_1}\) has length \(D_0\) in \(G^R\), there is a shortest path of length \(D_0\) from \(r\) to \(s_{e_1}\) in the subdivided graph. This path reaches \(s_{e_1}\) through one of the two edges replacing \(e_1\), and that edge has length \(w_{e_1}/2\). Hence the shortest path length from \(r\) to one endpoint of \(e_1\) in the original graph is at most \(D_0-w_{e_1}/2\), and reaching this endpoint and then traversing \(e_1\) has length at most \(D_0+w_{e_1}/2\le 2D_0\). For each \(i=1,\ldots,k-1\), since the edge \(s_{e_i}s_{e_{i+1}}\) has length \(D_i\) in \(G^R\), there is a shortest path of length \(D_i\) from \(s_{e_i}\) to \(s_{e_{i+1}}\) in the subdivided graph. By the construction of the subdivided graph, this path corresponds to a walk \(W_i\) in the original graph from an endpoint of \(e_i\) to an endpoint of \(e_{i+1}\), of length \(D_i-(w_{e_i}+w_{e_{i+1}})/2\). After traversing \(e_i\), the constructed walk is located at \(t_i\), which is an endpoint of \(e_i\). Thus the shortest path length from \(t_i\) to the endpoint of \(e_{i+1}\) where \(W_i\) ends is at most \(w_{e_i}\) plus the length of \(W_i\). Since the algorithm chooses an endpoint of \(e_{i+1}\) with minimum shortest path length from \(t_i\), the part of the constructed walk from \(t_i\) to \(t_{i+1}\) has length at most \(D_i-(w_{e_i}+w_{e_{i+1}})/2+w_{e_i}+w_{e_{i+1}} = D_i+(w_{e_i}+w_{e_{i+1}})/2\le 2D_i\). The part from \(t_k\) back to \(r\) is bounded in the same way and has length at most \(D_k+w_{e_k}/2\le 2D_k\). Therefore the length of the constructed walk is at most \(2(D_0+\cdots+D_k)=2L\). The constructed walk traverses every positive-profit edge whose representative vertex is visited by the returned solution to the PCTSP instance. Therefore the total profit of positive-profit edges not traversed by the constructed walk is at most the total penalty paid by the returned solution to the PCTSP instance. It follows that the objective value of the constructed walk is at most twice the objective value of the returned solution to the PCTSP instance. Since the returned solution to the PCTSP instance has objective value at most \(\rho\) times the optimum value of the reduced PCTSP instance, and this optimum value is at most the optimum value of the original PCRPP instance, the PCTSP-reduction algorithm is a \(2\rho\)-approximation algorithm for the PCRPP.

\section{Verification of the constant bound}
\label{app:boundanalysis}

\begin{lemma}
\label{lem:boundanalysis}
Let $\kappa_0=0.36621005$, $\kappa=0.99678328$, $\beta=1.98094420$, and let \(h_\xi(\kappa,\kappa_0)\) be the function defined in the proof of Theorem~\ref{thm:below_1_6}. Then
\[
\max_{\xi\in[\kappa_0,\kappa]}
\frac{h_\xi(\kappa,\kappa_0)}{1-\xi}
<1.59872206 .
\]
\end{lemma}

\begin{proof}
Let \(F(\xi)=h_\xi(\kappa,\kappa_0)/(1-\xi)\) and \(L=\kappa-\kappa_0\). For the above parameter values, \(0.63<L<0.631\), \(1-\kappa=0.00321672\), and \(1.9<\beta<2\). In particular, throughout \(\xi\in[\kappa_0,\kappa]\), we have \(1-\xi\ge 1-\kappa>0\), \(0\le \kappa-\xi\le L<1\), and all powers of \(\kappa-\xi\) and \(L\) appearing below are well defined and nonnegative. Also, \(L>0\), \(\nu>0\), and \(\beta+1,\beta+2>0\). We first derive the bound \(\nu/L<8\). By the definition of \(\nu\) in Theorem~\ref{thm:below_1_6},
\[
\begin{aligned}
L\nu^{-1}
&=
\frac{(3-\kappa)L^{\beta+2}}{\beta+1}
+
\frac{L^{\beta+3}}{\beta+2} \\
&>
\frac{2L^4}{3}
+
\frac{L^5}{4}
>
\frac{2(0.63)^4}{3}
+
\frac{(0.63)^5}{4}
>
\frac{1}{8} .
\end{aligned}
\]
Here we used \(3-\kappa>2\), \(\beta+1<3\), \(\beta+2<4\), and \(0<L<1\), so \(L^{\beta+2}>L^4\) and \(L^{\beta+3}>L^5\).

Taking the derivative of \(h_\xi(\kappa,\kappa_0)\) with respect to \(\xi\) gives
\[
\begin{aligned}
\left|\frac{d}{d\xi}h_\xi(\kappa,\kappa_0)\right|
&\le
\frac{\nu}{L}
\Bigg[
\left|
\bigl(\phi_\xi(\kappa_0)-\phi_\xi(\kappa)\bigr)
\frac{L^{\beta+2}-(\kappa-\xi)^{\beta+2}}{\beta+2}
+
\phi_\xi(\kappa)L
\frac{L^{\beta+1}-(\kappa-\xi)^{\beta+1}}{\beta+1}
\right| \\
&\quad
+\xi\Bigg|
\frac{d}{d\xi}\bigl(\phi_\xi(\kappa_0)-\phi_\xi(\kappa)\bigr)
\frac{L^{\beta+2}-(\kappa-\xi)^{\beta+2}}{\beta+2}
+
\bigl(\phi_\xi(\kappa_0)-\phi_\xi(\kappa)\bigr)
(\kappa-\xi)^{\beta+1} \\
&\quad\quad
+
\frac{d}{d\xi}\phi_\xi(\kappa)\,
L\frac{L^{\beta+1}-(\kappa-\xi)^{\beta+1}}{\beta+1}
+
\phi_\xi(\kappa)L(\kappa-\xi)^\beta
\Bigg|
\Bigg].
\end{aligned}
\]
As defined in the proof of Theorem~\ref{thm:below_1_6}, \(\phi_\xi(\delta)=(3-\delta-\kappa)(3-\delta)/(3-\delta-\xi)\). For \(\xi\in[\kappa_0,\kappa]\) and \(\delta\in\{\kappa_0,\kappa\}\), we have \(3-\delta-\xi>0\). Using \(\kappa_0=0.36621005\), \(\kappa=0.99678328\), \(\beta=1.98094420\), and \(0.63<L<0.631\), we have \(|\phi_\xi(\kappa_0)-\phi_\xi(\kappa)|<4.65\), \(\left|d(\phi_\xi(\kappa_0)-\phi_\xi(\kappa))/d\xi\right|<4\), \(0\le \phi_\xi(\kappa)<2.01\), \(0\le d\phi_\xi(\kappa)/d\xi<2\), \(0\le (L^{\beta+2}-(\kappa-\xi)^{\beta+2})/(\beta+2)<0.07\), \(0\le(\kappa-\xi)^{\beta+1}<0.399\), \(0\le (L^{\beta+1}-(\kappa-\xi)^{\beta+1})/(\beta+1)<0.138\), and \(0\le(\kappa-\xi)^\beta<0.631\). Substituting these bounds gives $\left|\frac{d}{d\xi}h_\xi(\kappa,\kappa_0)\right|<31$.

The definition of \(h_\xi(\kappa,\kappa_0)\) gives \(h_\kappa(\kappa,\kappa_0)=1-\kappa\). Since \(\left|d h_\xi(\kappa,\kappa_0)/d\xi\right|<31\) on \([\kappa_0,\kappa]\), for every \(\xi\in[\kappa_0,\kappa]\),
\[
|h_\xi(\kappa,\kappa_0)|
\le
|h_\kappa(\kappa,\kappa_0)|
+
\int_\xi^\kappa
\left|\frac{d}{ds}h_s(\kappa,\kappa_0)\right|\,ds
\le
1-\kappa+31(\kappa-\xi).
\]
Since \(1-\xi>0\) on \([\kappa_0,\kappa]\), and \(F'(\xi)=h'_\xi(\kappa,\kappa_0)/(1-\xi) +h_\xi(\kappa,\kappa_0)/(1-\xi)^2\), we obtain
\[
\begin{aligned}
|F'(\xi)|
&\le
\frac{31}{1-\xi}
+
\frac{1-\kappa+31(\kappa-\xi)}{(1-\xi)^2} =
\frac{62}{1-\xi}
-
\frac{30(1-\kappa)}{(1-\xi)^2}.
\end{aligned}
\]
The last expression is increasing in \(\xi\) on \([\kappa_0,\kappa]\), since its derivative is
\[
\frac{62(1-\xi)-60(1-\kappa)}{(1-\xi)^3}>0 \qquad
\forall\xi\in[\kappa_0,\kappa].
\]
Here the denominator is positive because \(1-\xi>0\), and the numerator is positive because \(1-\xi\ge 1-\kappa>0\). Therefore it is at most its value at \(\xi=\kappa\), namely \(32/(1-\kappa)\). Hence
\[
|F'(\xi)|\le \frac{32}{1-\kappa}
\qquad
\forall\xi\in[\kappa_0,\kappa].
\]

Now set \(\Delta=10^{-8}\) and $\Xi_\Delta
=
\{\kappa_0+j\Delta:\ j=0,1,\ldots,
\lfloor(\kappa-\kappa_0)/\Delta\rfloor\}\cup\{\kappa\}$. We evaluate \(F\) on all points of \(\Xi_\Delta\) using high-precision decimal arithmetic. The computation shows that \(\max_{\xi\in\Xi_\Delta}F(\xi)\) is attained at \(\xi=0.94817979\), where \(F(0.94817979)\approx 1.59862256\). Hence $\max_{\xi\in\Xi_\Delta}F(\xi)<1.59862257$. Since \(F\) is continuous on \([\kappa_0,\kappa]\), let \(\xi^*\in[\kappa_0,\kappa]\) be a point at which \(F\) attains its maximum. By the definition of \(\Xi_\Delta\), there exists \(\xi^\Delta\in\Xi_\Delta\) such that \(|\xi^\Delta-\xi^*|\le \Delta\). Using the derivative bound above, we obtain
\[
\begin{aligned}
\max_{\xi\in[\kappa_0,\kappa]}F(\xi)
&=F(\xi^*)\le F(\xi^\Delta)
+
\frac{32}{1-\kappa}|\xi^*-\xi^\Delta| \le \max_{\xi\in\Xi_\Delta}F(\xi)+ \frac{32}{1-\kappa}\Delta\\
&<
1.59862257+0.00009949=1.59872206.
\end{aligned}
\]
\end{proof}

\section{Detailed computational results}
\label{app:computational-results}

Table~\ref{tab:instance-results} reports the per-instance computational results. Here \(\mathrm{OPT}\) is the optimum value for the minimization objective, \(\mathrm{ALG}\) is the objective value returned by Algorithm~\ref{alg:best_of_many}, \(\mathrm{RED}\) is the objective value returned by the PCTSP-reduction algorithm using the PCTSP algorithm of \citet{Blauth2026}, and \(\mathrm{OPT}_{\mathrm{LP}}\) is the optimum value of \textsc{PCRPP-LP} on the preprocessed complete graph. The running times are in seconds. For each instance, the better value between \(\mathrm{ALG}\) and \(\mathrm{RED}\) is shown in bold; if the two values are equal, both are shown in bold.

\scriptsize
\setlength{\tabcolsep}{3pt}

\begin{longtable}{lrrrrrrrrrrr}
\caption{Per-instance computational results.}
\label{tab:instance-results}\\
\toprule
Instance
& \(|V|\)
& \(|E|\)
& \(\mathrm{OPT}\)
& \(\mathrm{ALG}\)
& \(\mathrm{RED}\)
& \(\mathrm{OPT}_{\mathrm{LP}}\)
& ALG gap
& RED gap
& LP gap
& ALG time
& RED time \\
\midrule
\endfirsthead

\toprule
Instance
& \(|V|\)
& \(|E|\)
& \(\mathrm{OPT}\)
& \(\mathrm{ALG}\)
& \(\mathrm{RED}\)
& \(\mathrm{OPT}_{\mathrm{LP}}\)
& ALG gap
& RED gap
& LP gap
& ALG time
& RED time \\
\midrule
\endhead

\midrule
\multicolumn{12}{r}{Continued on next page}\\
\midrule
\endfoot

\bottomrule
\endlastfoot

\midrule
\multicolumn{12}{l}{\textbf{ALBAIDA}} \\
ALBAIDAANoRPP & 102 & 160 & 11717 & \textbf{12530} & 13098 & 11296 & 6.94 & 11.79 & 3.59 & 2717.35 & 1176.25 \\
ALBAIDABNoRPP & 90 & 144 & 10502 & \textbf{10961} & 11637 & 10008 & 4.37 & 10.81 & 4.70 & 1957.28 & 366.60 \\

\midrule
\multicolumn{12}{l}{\textbf{D}} \\
D00 & 16 & 32 & 1133 & \textbf{1133} & \textbf{1133} & 1133 & 0.00 & 0.00 & 0.00 & 0.41 & 0.20 \\
D01 & 16 & 31 & 1278 & \textbf{1278} & 1329 & 1278 & 0.00 & 3.99 & 0.00 & 2.31 & 0.61 \\
D02 & 16 & 31 & 1158 & \textbf{1185} & \textbf{1185} & 1154 & 2.33 & 2.33 & 0.35 & 2.91 & 0.80 \\
D03 & 16 & 32 & 1304 & \textbf{1304} & 1343 & 1296 & 0.00 & 2.99 & 0.61 & 1.35 & 0.44 \\
D04 & 16 & 31 & 1395 & 1429 & \textbf{1427} & 1389 & 2.44 & 2.29 & 0.43 & 2.24 & 0.64 \\
D05 & 16 & 31 & 1252 & \textbf{1252} & 1282 & 1220 & 0.00 & 2.40 & 2.56 & 1.69 & 0.57 \\
D06 & 16 & 32 & 1402 & \textbf{1474} & 1534 & 1387 & 5.14 & 9.42 & 1.07 & 4.23 & 0.87 \\
D07 & 16 & 31 & 1413 & 1502 & \textbf{1461} & 1369 & 6.30 & 3.40 & 3.11 & 4.31 & 1.03 \\
D08 & 16 & 31 & 1383 & \textbf{1445} & 1562 & 1323 & 4.48 & 12.94 & 4.34 & 4.42 & 1.16 \\
D09 & 36 & 72 & 1814 & \textbf{1814} & 1879 & 1814 & 0.00 & 3.58 & 0.00 & 16.42 & 9.69 \\
D10 & 36 & 72 & 1765 & \textbf{1765} & \textbf{1765} & 1765 & 0.00 & 0.00 & 0.00 & 3.37 & 5.35 \\
D11 & 36 & 72 & 2207 & \textbf{2325} & 2433 & 2177 & 5.35 & 10.24 & 1.36 & 46.08 & 10.42 \\
D12 & 36 & 72 & 2031 & \textbf{2061} & 2308 & 2019.5 & 1.48 & 13.64 & 0.57 & 44.90 & 15.18 \\
D13 & 36 & 72 & 2048 & 2177 & \textbf{2157} & 2031 & 6.30 & 5.32 & 0.83 & 42.54 & 10.25 \\
D14 & 36 & 72 & 2153 & \textbf{2249} & 2325 & 2101 & 4.46 & 7.99 & 2.42 & 102.19 & 14.59 \\
D15 & 36 & 72 & 2211 & 2410 & \textbf{2325} & 2148 & 9.00 & 5.16 & 2.85 & 112.27 & 18.33 \\
D16 & 36 & 72 & 2095 & \textbf{2175} & 2241 & 2014 & 3.82 & 6.97 & 3.87 & 123.40 & 20.73 \\
D17 & 36 & 72 & 2369 & \textbf{2598} & 2625 & 2265 & 9.67 & 10.81 & 4.39 & 125.30 & 20.61 \\
D18 & 64 & 128 & 2151 & \textbf{2185} & 2375 & 2113.5 & 1.58 & 10.41 & 1.74 & 514.57 & 120.61 \\
D19 & 64 & 128 & 2293 & \textbf{2393} & 2510 & 2265 & 4.36 & 9.46 & 1.22 & 700.89 & 162.67 \\
D20 & 64 & 128 & 2349 & \textbf{2378} & 2427 & 2312 & 1.23 & 3.32 & 1.58 & 209.81 & 89.23 \\
D21 & 64 & 128 & 2529 & \textbf{2628} & 2904 & 2468 & 3.91 & 14.83 & 2.41 & 623.29 & 250.34 \\
D22 & 64 & 128 & 2610 & \textbf{2793} & 2873 & 2532 & 7.01 & 10.08 & 2.99 & 626.20 & 178.64 \\
D23 & 64 & 128 & 2578 & \textbf{2775} & 2851 & 2440 & 7.64 & 10.59 & 5.35 & 567.41 & 170.54 \\
D24 & 64 & 128 & 2733 & \textbf{2931} & 3059 & 2568 & 7.24 & 11.93 & 6.04 & 738.15 & 193.12 \\
D25 & 64 & 128 & 2884 & \textbf{3023} & 3381 & 2807 & 4.82 & 17.23 & 2.67 & 737.80 & 295.49 \\
D26 & 64 & 128 & 3129 & \textbf{3353} & 3516 & 2939 & 7.16 & 12.37 & 6.07 & 850.13 & 188.82 \\
D27 & 100 & 200 & 3131 & \textbf{3232} & 3245 & 3060 & 3.23 & 3.64 & 2.27 & 3157.33 & 930.56 \\
D28 & 100 & 200 & 3480 & \textbf{3647} & 3751 & 3400 & 4.80 & 7.79 & 2.30 & 2728.59 & 632.08 \\
D29 & 100 & 200 & 2943 & \textbf{3155} & 3247 & 2874 & 7.20 & 10.33 & 2.34 & 3913.36 & 2735.59 \\
D30 & 100 & 200 & 3367 & \textbf{3530} & 3827 & 3264 & 4.84 & 13.66 & 3.06 & 4728.08 & 4506.35 \\
D31 & 100 & 200 & 3670 & \textbf{3870} & 4154 & 3487 & 5.45 & 13.19 & 4.99 & 4073.35 & 832.13 \\
D32 & 100 & 200 & 3311 & \textbf{3490} & 3564 & 3208 & 5.41 & 7.64 & 3.11 & 4163.01 & 1365.85 \\
D33 & 100 & 200 & 3869 & \textbf{4110} & 4674 & 3759.5 & 6.23 & 20.81 & 2.83 & 6127.37 & 2817.81 \\
D34 & 100 & 200 & 4012 & \textbf{4228} & 4426 & 3859 & 5.38 & 10.32 & 3.81 & 3575.13 & 903.64 \\
D35 & 100 & 200 & 3737 & \textbf{3886} & 4266 & 3561.5 & 3.99 & 14.16 & 4.70 & 8295.91 & 8159.58 \\

\midrule
\multicolumn{12}{l}{\textbf{G}} \\
G00 & 16 & 24 & 6 & \textbf{6} & \textbf{6} & 6 & 0.00 & 0.00 & 0.00 & 0.04 & 0.00 \\
G01 & 16 & 24 & 11 & \textbf{11} & \textbf{11} & 11 & 0.00 & 0.00 & 0.00 & 0.04 & 0.01 \\
G02 & 16 & 24 & 9 & \textbf{9} & \textbf{9} & 9 & 0.00 & 0.00 & 0.00 & 0.04 & 0.01 \\
G03 & 16 & 24 & 15 & \textbf{15} & \textbf{15} & 15 & 0.00 & 0.00 & 0.00 & 0.11 & 0.02 \\
G04 & 16 & 24 & 16 & \textbf{16} & 18 & 16 & 0.00 & 12.50 & 0.00 & 0.10 & 0.04 \\
G05 & 16 & 24 & 12 & \textbf{12} & 14 & 12 & 0.00 & 16.67 & 0.00 & 0.19 & 0.04 \\
G06 & 16 & 24 & 17 & \textbf{17} & 20 & 17 & 0.00 & 17.65 & 0.00 & 0.37 & 0.16 \\
G07 & 16 & 24 & 16 & \textbf{16} & \textbf{16} & 16 & 0.00 & 0.00 & 0.00 & 0.19 & 0.07 \\
G08 & 16 & 24 & 15 & \textbf{15} & 18 & 15 & 0.00 & 20.00 & 0.00 & 0.17 & 0.08 \\
G09 & 36 & 60 & 21 & \textbf{21} & \textbf{21} & 21 & 0.00 & 0.00 & 0.00 & 0.42 & 0.06 \\
G10 & 36 & 60 & 26 & \textbf{26} & 28 & 26 & 0.00 & 7.69 & 0.00 & 0.24 & 0.17 \\
G11 & 36 & 60 & 26 & \textbf{26} & 28 & 26 & 0.00 & 7.69 & 0.00 & 0.87 & 0.21 \\
G12 & 36 & 60 & 40 & \textbf{42} & 45 & 38 & 5.00 & 12.50 & 5.00 & 6.35 & 1.17 \\
G13 & 36 & 60 & 37 & \textbf{38} & 41 & 35 & 2.70 & 10.81 & 5.41 & 5.90 & 0.91 \\
G14 & 36 & 60 & 39 & 43 & \textbf{42} & 37 & 10.26 & 7.69 & 5.13 & 7.98 & 1.24 \\
G15 & 36 & 60 & 45 & \textbf{48} & 54 & 43 & 6.67 & 20.00 & 4.44 & 15.95 & 3.65 \\
G16 & 36 & 60 & 41 & \textbf{41} & 52 & 40.5 & 0.00 & 26.83 & 1.22 & 9.81 & 3.70 \\
G17 & 36 & 60 & 44 & \textbf{45} & 50 & 43 & 2.27 & 13.64 & 2.27 & 7.61 & 4.93 \\
G18 & 64 & 112 & 43 & \textbf{43} & 45 & 43 & 0.00 & 4.65 & 0.00 & 9.93 & 0.94 \\
G19 & 64 & 112 & 50 & \textbf{53} & \textbf{53} & 47.5 & 6.00 & 6.00 & 5.00 & 12.83 & 1.21 \\
G20 & 64 & 112 & 47 & \textbf{47} & 51 & 47 & 0.00 & 8.51 & 0.00 & 6.68 & 0.99 \\
G21 & 64 & 112 & 64 & \textbf{68} & 74 & 60.5 & 6.25 & 15.62 & 5.47 & 57.86 & 21.25 \\
G22 & 64 & 112 & 66 & \textbf{74} & 77 & 64.5 & 12.12 & 16.67 & 2.27 & 50.01 & 10.97 \\
G23 & 64 & 112 & 70 & \textbf{75} & 80 & 67 & 7.14 & 14.29 & 4.29 & 67.17 & 13.78 \\
G24 & 64 & 112 & 83 & \textbf{91} & 99 & 79.5 & 9.64 & 19.28 & 4.22 & 151.19 & 72.34 \\
G25 & 64 & 112 & 77 & \textbf{86} & \textbf{86} & 72.5 & 11.69 & 11.69 & 5.84 & 88.35 & 34.48 \\
G26 & 64 & 112 & 78 & \textbf{87} & 93 & 76 & 11.54 & 19.23 & 2.56 & 113.49 & 39.99 \\
G27 & 100 & 180 & 71 & \textbf{73} & 75 & 70 & 2.82 & 5.63 & 1.41 & 65.62 & 4.07 \\
G28 & 100 & 180 & 79 & \textbf{87} & 89 & 77.5 & 10.13 & 12.66 & 1.90 & 75.91 & 17.87 \\
G29 & 100 & 180 & 77 & \textbf{84} & 87 & 75.5 & 9.09 & 12.99 & 1.95 & 66.65 & 12.87 \\
G30 & 100 & 180 & 100 & \textbf{107} & 111 & 97.25 & 7.00 & 11.00 & 2.75 & 380.47 & 40.54 \\
G31 & 100 & 180 & 103 & \textbf{110} & 113 & 100.38 & 6.80 & 9.71 & 2.55 & 329.95 & 69.14 \\
G32 & 100 & 180 & 110 & \textbf{116} & 124 & 104 & 5.45 & 12.73 & 5.45 & 336.04 & 93.92 \\
G33 & 100 & 180 & 136 & \textbf{148} & 151 & 126.25 & 8.82 & 11.03 & 7.17 & 1198.95 & 233.98 \\
G34 & 100 & 180 & 130 & \textbf{139} & 153 & 122 & 6.92 & 17.69 & 6.15 & 596.57 & 305.46 \\
G35 & 100 & 180 & 132 & \textbf{137} & 151 & 124.5 & 3.79 & 14.39 & 5.68 & 913.49 & 294.66 \\

\midrule
\multicolumn{12}{l}{\textbf{P}} \\
P01 & 11 & 13 & 60 & \textbf{60} & \textbf{60} & 60 & 0.00 & 0.00 & 0.00 & 0.09 & 0.03 \\
P02 & 14 & 33 & 293 & \textbf{306} & 307 & 291 & 4.44 & 4.78 & 0.68 & 2.16 & 0.76 \\
P03 & 28 & 58 & 134 & \textbf{145} & 149 & 132 & 8.21 & 11.19 & 1.49 & 20.86 & 5.15 \\
P04 & 17 & 35 & 105 & \textbf{106} & 125 & 99 & 0.95 & 19.05 & 5.71 & 5.56 & 1.57 \\
P05 & 20 & 35 & 197 & \textbf{197} & 209 & 197 & 0.00 & 6.09 & 0.00 & 3.13 & 1.31 \\
P06 & 24 & 46 & 130 & \textbf{130} & 147 & 125 & 0.00 & 13.08 & 3.85 & 9.00 & 2.42 \\
P07 & 23 & 47 & 182 & \textbf{182} & 188 & 171.5 & 0.00 & 3.30 & 5.77 & 16.72 & 3.10 \\
P08 & 17 & 40 & 166 & \textbf{171} & 189 & 157 & 3.01 & 13.86 & 5.42 & 6.57 & 2.33 \\
P09 & 14 & 26 & 100 & \textbf{101} & 105 & 95 & 1.00 & 5.00 & 5.00 & 2.02 & 0.63 \\
P10 & 12 & 20 & 95 & \textbf{103} & 104 & 92 & 8.42 & 9.47 & 3.16 & 0.99 & 0.32 \\
P11 & 9 & 14 & 27 & \textbf{27} & 31 & 27 & 0.00 & 14.81 & 0.00 & 0.17 & 0.11 \\
P12 & 7 & 18 & 23 & \textbf{23} & 25 & 23 & 0.00 & 8.70 & 0.00 & 0.19 & 0.07 \\
P13 & 7 & 10 & 40 & \textbf{40} & 41 & 40 & 0.00 & 2.50 & 0.00 & 0.13 & 0.02 \\
P14 & 28 & 79 & 417 & \textbf{446} & 451 & 406 & 6.95 & 8.15 & 2.64 & 54.06 & 14.70 \\
P15 & 26 & 37 & 479 & \textbf{484} & 534 & 470 & 1.04 & 11.48 & 1.88 & 10.03 & 3.08 \\
P16 & 31 & 94 & 409 & \textbf{416} & 444 & 403.5 & 1.71 & 8.56 & 1.34 & 124.88 & 33.62 \\
P17 & 19 & 44 & 209 & \textbf{214} & 222 & 207 & 2.39 & 6.22 & 0.96 & 6.61 & 4.26 \\
P18 & 23 & 37 & 183 & \textbf{189} & 204 & 179 & 3.28 & 11.48 & 2.19 & 4.66 & 1.49 \\
P19 & 33 & 54 & 322 & \textbf{338} & 352 & 311 & 4.97 & 9.32 & 3.42 & 30.79 & 13.18 \\
P20 & 50 & 98 & 622 & \textbf{634} & 665 & 598 & 1.93 & 6.91 & 3.86 & 348.43 & 85.28 \\
P21 & 49 & 110 & 542 & \textbf{573} & 632 & 524 & 5.72 & 16.61 & 3.32 & 392.55 & 121.29 \\
P22 & 50 & 184 & 1595 & \textbf{1641} & 1732 & 1561 & 2.88 & 8.59 & 2.13 & 1291.73 & 578.36 \\
P23 & 50 & 158 & 856 & \textbf{913} & 991 & 838 & 6.66 & 15.77 & 2.10 & 1283.47 & 319.59 \\
P24 & 41 & 125 & 840 & \textbf{844} & 893 & 817 & 0.48 & 6.31 & 2.74 & 291.15 & 112.78 \\

\midrule
\multicolumn{12}{l}{\textbf{R}} \\
R00 & 20 & 37 & 44158 & \textbf{44158} & 48290 & 44158 & 0.00 & 9.36 & 0.00 & 1.78 & 0.82 \\
R01 & 20 & 47 & 68689 & \textbf{68689} & 75165 & 68506 & 0.00 & 9.43 & 0.27 & 7.04 & 4.69 \\
R02 & 20 & 47 & 61951 & \textbf{61951} & 62975 & 61951 & 0.00 & 1.65 & 0.00 & 4.97 & 2.17 \\
R03 & 20 & 75 & 131531 & \textbf{134619} & 134978 & 131461 & 2.35 & 2.62 & 0.05 & 29.00 & 9.64 \\
R04 & 20 & 60 & 83475 & \textbf{83888} & 90903 & 82966 & 0.49 & 8.90 & 0.61 & 7.74 & 3.62 \\
R05 & 30 & 70 & 72481 & \textbf{72481} & 73960 & 72481 & 0.00 & 2.04 & 0.00 & 44.46 & 30.95 \\
R06 & 30 & 112 & 132498 & \textbf{133900} & 137718 & 132371 & 1.06 & 3.94 & 0.10 & 405.94 & 261.22 \\
R07 & 30 & 70 & 78598 & \textbf{79071} & 81447 & 78125 & 0.60 & 3.62 & 0.60 & 43.06 & 11.16 \\
R08 & 30 & 111 & 134271 & \textbf{138891} & 143226 & 134142.5 & 3.44 & 6.67 & 0.10 & 220.66 & 102.77 \\
R09 & 30 & 111 & 136052 & \textbf{136052} & 148485 & 136052 & 0.00 & 9.14 & 0.00 & 88.86 & 41.55 \\
R10 & 40 & 130 & 114125 & \textbf{115409} & 117771 & 113368 & 1.13 & 3.19 & 0.66 & 134.10 & 189.21 \\
R11 & 40 & 103 & 89530 & \textbf{92304} & 94376 & 89439 & 3.10 & 5.41 & 0.10 & 113.93 & 70.12 \\
R12 & 40 & 82 & 78861 & \textbf{78861} & 83773 & 78742 & 0.00 & 6.23 & 0.15 & 40.75 & 13.95 \\
R13 & 40 & 203 & 259628 & \textbf{263882} & 273048 & 258719 & 1.64 & 5.17 & 0.35 & 553.21 & 366.11 \\
R14 & 40 & 203 & 251283 & \textbf{256939} & 273369 & 250803 & 2.25 & 8.79 & 0.19 & 714.29 & 508.17 \\
R15 & 50 & 203 & 204271 & \textbf{204326} & 215354 & 203783 & 0.03 & 5.43 & 0.24 & 1109.57 & 648.26 \\
R16 & 50 & 162 & 133707 & \textbf{134072} & 135973 & 133633 & 0.27 & 1.69 & 0.06 & 2193.78 & 668.35 \\
R17 & 50 & 130 & 108586 & \textbf{110964} & 113860 & 108446 & 2.19 & 4.86 & 0.13 & 399.51 & 208.52 \\
R18 & 50 & 203 & 197692 & \textbf{199417} & 202800 & 197399 & 0.87 & 2.58 & 0.15 & 874.01 & 1118.60 \\
R19 & 50 & 203 & 197364 & \textbf{198718} & 203493 & 197145 & 0.69 & 3.11 & 0.11 & 548.86 & 822.07 \\

\end{longtable}

%
%


\begin{thebibliography}{99}

\bibitem[Ar\'aoz et al.(2006)]{Araoz2006}
Ar\'aoz, J., Fern\'andez, E., \& Zoltan, C. (2006). Privatized rural postman problems. \textit{Computers \& Operations Research}, \textit{33}(12), 3432--3449. \url{https://doi.org/10.1016/j.cor.2005.02.013}.

\bibitem[Ar\'aoz et al.(2009)]{Araoz2009}
Ar\'aoz, J., Fern\'andez, E., \& Meza, O. (2009). Solving the prize-collecting rural postman problem. \textit{European Journal of Operational Research}, \textit{196}(3), 886--896. \url{https://doi.org/10.1016/j.ejor.2008.04.037}.

\bibitem[Ar\'aoz et al.(2009)]{AraozFranquesa2009}
Ar\'aoz, J., Fern\'andez, E., \& Franquesa, C. (2009). The clustered prize-collecting arc routing problem. \textit{Transportation Science}, \textit{43}(3), 287--300. \url{https://doi.org/10.1287/trsc.1090.0270}.

\bibitem[Archetti and Speranza(2015)]{ArchettiSperanza2015}
Archetti, C., \& Speranza, M. G. (2015). Arc routing problems with profits. In \'A. Corber\'an \& G. Laporte (Eds.), \textit{Arc Routing: Problems, Methods, and Applications} (pp. 281--299). SIAM. \url{https://doi.org/10.1137/1.9781611973679.ch12}.

\bibitem[Archetti et al.(2014)]{ArchettiGuastarobaSperanza2014}
Archetti, C., Guastaroba, G., \& Speranza, M. G. (2014). An ILP-refined tabu search for the directed profitable rural postman problem. \textit{Discrete Applied Mathematics}, \textit{163}, 3--16. \url{https://doi.org/10.1016/j.dam.2012.06.002}.

\bibitem[\'Avila et al.(2016)]{Avila2016}
\'Avila, T., Corber\'an, \'A., Plana, I., \& Sanchis, J. M. (2016). A branch-and-cut algorithm for the profitable windy rural postman problem. \textit{European Journal of Operational Research}, \textit{249}(3), 1092--1101. \url{https://doi.org/10.1016/j.ejor.2015.10.016}.

\bibitem[Balas(1989)]{Balas1989}
Balas, E. (1989). The prize-collecting traveling salesman problem. \textit{Networks}, \textit{19}(6), 621--636. \url{https://doi.org/10.1002/net.3230190602}.

\bibitem[Bienstock et al.(1993)]{Bienstock1993}
Bienstock, D., Goemans, M. X., Simchi-Levi, D., \& Williamson, D. (1993). A note on the prize-collecting traveling salesman problem. \textit{Mathematical Programming}, \textit{59}(1), 413--420. \url{https://doi.org/10.1007/BF01581256}.

\bibitem[Black et al.(2013)]{BlackEgleseWohlk2013}
Black, D., Eglese, R., \& W{\o}hlk, S. (2013).
The time-dependent prize-collecting arc routing problem.
\textit{Computers \& Operations Research}, \textit{40}(2), 526--535.
\url{https://doi.org/10.1016/j.cor.2012.08.001}.

\bibitem[Blauth and N\"agele(2023)]{Blauth2023}
Blauth, J., \& N\"agele, M. (2023). An improved approximation guarantee for prize-collecting TSP. In \textit{Proceedings of the 55th Annual ACM Symposium on Theory of Computing} (pp. 1848--1861). ACM. \url{https://doi.org/10.1145/3564246.3585159}.

\bibitem[Blauth et al.(2026)]{Blauth2026}
Blauth, J., Klein, N., \& N\"agele, M. (2026).
A better-than-1.6-approximation for prize-collecting TSP.
\textit{Mathematical Programming}, \textit{216}(1--2), 87--109.
\url{https://doi.org/10.1007/s10107-025-02221-4}.

\bibitem[Colombi and Mansini(2014)]{ColombiMansini2014}
Colombi, M., \& Mansini, R. (2014). New results for the directed profitable rural postman problem. \textit{European Journal of Operational Research}, \textit{238}(3), 760--773. \url{https://doi.org/10.1016/j.ejor.2014.05.006}.

\bibitem[Corber\'an et al.(2011)]{CorberanFernandezFranquesaSanchis2011}
Corber\'an, A., Fern\'andez, E., Franquesa, C., \& Sanchis, J. M. (2011).
The windy clustered prize-collecting arc-routing problem.
\textit{Transportation Science}, \textit{45}(3), 317--334.
\url{https://doi.org/10.1287/trsc.1110.0370}.

\bibitem[Edmonds and Johnson(1973)]{Edmonds1973}
Edmonds, J., \& Johnson, E. L. (1973). Matching, Euler tours and the Chinese postman. \textit{Mathematical Programming}, \textit{5}, 88--124. \url{https://doi.org/10.1007/BF01580113}.

\bibitem[Frank(1992)]{Frank1992}
Frank, A. (1992). On a theorem of Mader. \textit{Discrete Mathematics}, \textit{101}, 49--57. \url{https://doi.org/10.1016/0012-365X(92)90589-8}.

\bibitem[Goemans and Williamson(1995)]{GoemansWilliamson1995}
Goemans, M. X., \& Williamson, D. P. (1995). A general approximation technique for constrained forest problems. \textit{SIAM Journal on Computing}, \textit{24}(2), 296--317. \url{https://doi.org/10.1137/S0097539793242618}.

\bibitem[Goemans(2009)]{Goemans2009}
Goemans, M. X. (2009). Combining approximation algorithms for the prize-collecting TSP. \textit{arXiv preprint arXiv:0910.0553}. \url{https://arxiv.org/abs/0910.0553}.

\bibitem[Gr\"otschel et al.(1981)]{Grotschel1981}
Gr\"otschel, M., Lov\'asz, L., \& Schrijver, A. (1981).
The ellipsoid method and its consequences in combinatorial optimization.
\textit{Combinatorica}, \textit{1}(2), 169--197.
\url{https://doi.org/10.1007/BF02579273}.

\bibitem[Khachiyan(1979)]{Khachiyan1979}
Khachiyan, L. G. (1979).
A polynomial algorithm in linear programming (English translation).
\textit{Soviet Mathematics Doklady}, \textit{20}, 191--194.

\bibitem[Lov\'asz(1976)]{Lovasz1976}
Lov\'asz, L. (1976). On some connectivity properties of Eulerian graphs. \textit{Acta Mathematica Academiae Scientiarum Hungarica}, \textit{28}, 129--138. \url{https://doi.org/10.1007/BF01902503}.

\bibitem[Mader(1978)]{Mader1978}
Mader, W. (1978). A reduction method for edge-connectivity in graphs. \textit{Annals of Discrete Mathematics}, \textit{3}, 145--164. \url{https://doi.org/10.1016/S0167-5060(08)70504-1}.

\bibitem[Orloff(1974)]{Orloff1974}
Orloff, C. S. (1974). A fundamental problem in vehicle routing.
\textit{Networks}, \textit{4}(1), 35--64.
\url{https://doi.org/10.1002/net.3230040105}.

\bibitem[Palma(2011)]{Palma2011}
Palma, G. (2011). A tabu search heuristic for the prize-collecting rural postman problem. \textit{Electronic Notes in Theoretical Computer Science}, \textit{281}, 85--100. \url{https://doi.org/10.1016/j.entcs.2011.11.027}.

\bibitem[Pan and Zhu(2024)]{PanZhu2024}
Pan, P., \& Zhu, H. (2024). Approximation algorithms for the restricted \(k\)-Chinese postman problems with penalties. \textit{Optimization Letters}, \textit{18}, 307--318. \url{https://doi.org/10.1007/s11590-023-01992-z}.

\bibitem[Zhu and Pan(2021)]{ZhuPan2021}
Zhu, H., \& Pan, P. (2021). The restricted Chinese postman problems with penalties. \textit{Operations Research Letters}, \textit{49}(6), 851--854. \url{https://doi.org/10.1016/j.orl.2021.10.002}.
\end{thebibliography}
\end{document}